\documentclass[11pt]{llncs}
\usepackage{times}

\usepackage{latexsym,amsmath,amssymb,url}

\newcommand{\Prob}{{\mathbb P}}
\newcommand{\Exp}{{\mathbb E}}
\newcommand{\kap}{\mbox{$\kappa$}}

\newtheorem{krule}{Reduction Rule}

\newcommand{\pp}{\hat{\phi}}

\begin{document}

\title{Constraint Satisfaction Problems Parameterized Above or Below Tight Bounds: A Survey}

\author{ Gregory Gutin and Anders Yeo}

\authorrunning{Gutin and Yeo}

\institute{
Royal Holloway, University of London, United Kingdom \\ \email{\{gutin,anders\}@cs.rhul.ac.uk}
}

\date{}
\maketitle

\begin{center}
{\bf This paper is dedicated to the 60th Birthday of Michael R. Fellows}
\end{center}

\begin{abstract}
\noindent
We consider constraint satisfaction problems parameterized above or below tight bounds. One example is MaxSat parameterized above $m/2$:
given a CNF formula $F$ with $m$ clauses, decide whether there is a truth assignment that satisfies at least $m/2+k$ clauses, where $k$ is the parameter.
Among other problems we deal with are MaxLin2-AA (given a system of linear equations over $\mathbb{F}_2$ in which each equation has a positive integral weight, decide whether there is an assignment to the variables that satisfies equations of total weight at least $W/2+k$, where $W$ is the total weight of all equations), Max-$r$-Lin2-AA (the same as MaxLin2-AA, but each equation has at most $r$ variables, where $r$ is a constant) and Max-$r$-Sat-AA (given a CNF formula $F$ with $m$ clauses in which each clause has at most $r$ literals, decide whether there is a truth assignment satisfying at least $\sum_{i=1}^m(1-2^{r_i})+k$ clauses, where $k$ is the parameter, $r_i$ is the number of literals in Clause $i$, and $r$ is a constant). We also consider
Max-$r$-CSP-AA, a natural generalization of both Max-$r$-Lin2-AA and Max-$r$-Sat-AA, order (or, permutation) constraint satisfaction problems of arities 2 and 3 parameterized above the average value and some other problems related to MaxSat. We discuss results, both polynomial kernels and parameterized algorithms, obtained for the problems mainly in the last few years as well as some open questions.
\end{abstract}

\pagenumbering{arabic}
\pagestyle{plain}

\section{Introduction}\label{section:intro}

This paper surveys mainly recent results in a subarea of parameterized algorithms and complexity that was launched quite early in the short history
of parameterized algorithms and complexity, namely, in the Year 2 BDF\footnote{BDF stands for Before Downey-Fellows, i.e., before 1999 when the first monograph describing foundations of parameterized algorithms and complexity was published \cite{DowneyFellows99}.}.

Consider  the well-known problem {\sc MaxSat}, where for a given CNF formula $F$ with $m$ clauses, we are asked to
determine the maximum number of clauses of $F$ that can be satisfied by a truth assignment.
It is well-known (and shown below, in Section \ref{sec:sat}) that there exists a truth assignment to the variables of
$F$ which satisfies at least $m/2$ clauses.

The standard parametrization $k$-{\sc MaxSat} of {\sc MaxSat} is as follows: decide whether there is a truth assignment
which satisfies at least $k$ clauses of $F$, where $k$ is the parameter.
(We provide basic terminology and notation on parameterized algorithms and complexity in Section \ref{sec:tn}.)
It is very easy to see that $k$-{\sc MaxSat} has a kernel with a linear number of variables. Indeed, consider an instance $I$
of $k$-{\sc MaxSat}. If $k\le m/2$ then $I$ is a {\sc Yes}-instance. Otherwise, we have $k>m/2$ and $m\le 2k-1.$
Suppose that we managed somehow to obtain a better result, a kernel with at most $pk$ variables, where $1\le p<2$. Is such a kernel of any interest?
Such a kernel would be of interest only for $k> m/2$, i.e., when the size of the kernel
would be bounded by $pk> pm/2$. Thus, such a kernel should be viewed as {\em huge} rather than {\em small} as the bound $pk$ might suggest
at the first glance.

The bound $m/2$ is tight as we can satisfy only half clauses in the instances consisting of pairs $(x),(\bar{x})$ of clauses. This suggest
the following parameterization of {\sc MaxSat} {\em above tight bound} introduced by Mahajan and Raman \cite{Mahajan97}:
decide whether there is a truth assignment which satisfies at least $m/2+k$ clauses of $F$, where $k$ is the parameter.

To the best our knowledge, \cite{Mahajan97} was the first paper on problems parameterized above or below tight bounds
and remained the only one for several years, at least for constraint satisfaction problems (CSPs). However, in the last few years
the study of CSPs parameterized above or below tight bounds has finally picked up. This is, in large part, due to
emergence of new probabilistic and linear-algebraic methods and approaches in the area.

In this survey paper, we will overview several results on CSPs parameterized above or below tight bounds, as well as some methods used to obtain these
results. While not going into details of the proofs, we will discuss some ideas behind the proofs.
We will also consider several open problems in the area.

In the remainder of this section we give a brief overview of the paper and its organization.

In the next section we provide basics on parameterized algorithms and complexity. The notions
mentioned there are all well-known apart from a recent notion of a bikernel introduced by Alon et al. \cite{AlonEtAl2009a}.
In Section \ref{sec:prob}, we describe some probabilistic and Harmonic Analysis tools. These tools are, in particular, used in
the recently introduced Strictly-Above-Below-Expectation method \cite{GutinKimSzeiderYeo09a}.

Results on {\sc MaxSat} parameterized above or below tight bounds are discussed in Section \ref{sec:sat}.
We will consider the above-mentioned parameterization of {\sc MaxSat} above tight bound, some ``stronger" parameterizations of {\sc MaxSat}
introduced or inspired by Mahajan and Raman \cite{Mahajan97}. The stronger parameterizations are based on the notion of a $t$-satisfiable CNF formula
(a formula in which each set of $t$ clauses can be satisfied by a truth assignment)
and asymptotically tight lower bounds on the maximum number of clauses of a $t$-satisfiable CNF formula satisfied by a truth assignment for $t=2$ and 3.
We will describe linear-variable kernels obtained for both $t=2$ and 3.

We will also consider the parameterization of {\sc 2-Sat} below the upper bound $m$, the number of clauses. This problem was proved to be fixed-parameter tractable by Razgon and O'Sullivan \cite{RazOsu}. Raman et al. \cite{RamRamSau} and Cygan et al. \cite{CygPilPilWoj} designed faster parameterized algorithms
for the problem. The problem has several application, which we will briefly overview.

Boolean Maximum $r$-CSPs parameterized above the average value are considered in Section \ref{sec:csp}, where $r$ is a positive integral constant.
In general, the Maximum $r$-CSP is given by a set $V$ of $n$ variables
and a set of $m$ Boolean formulas; each formula is assigned an integral positive weight and contains at most $r$ variables from $V$. The aim is to find a truth assignment which maximizes the weight of satisfied formulas. Averaging over all truth assignments, we can find the average value $A$ of the weight of satisfied formulas. It is easy to show that we can always find a truth assignment to the variables of $V$ which satisfied formulas of total weight at least $A$.
Thus, a natural parameterized problem is whether there exists a truth assignment that satisfies formulas of total weight at least $A+k,$ where $k$ is the parameter ($k$ is a nonnegative integer). We denote such a problem by {\sc Max-$r$-CSP-AA}.

The problem {\sc Max-$r$-Lin2-AA} is a special case of {\sc Max-$r$-CSP-AA} when every formula is a linear equation over
$\mathbb{F}_2$ with at most $r$ variables. For {\sc Max-$r$-Lin2-AA}, we have $A=W/2$, where $W$ is the total weight of all equations.
It is well-known that, in polynomial time, we can find an assignment to the variables that satisfies equations of total weight at least $W/2$, but, for any $\epsilon > 0$ it is NP-hard to decide whether there is an assignment satisfying equations of total weight at least $W(1+\epsilon)/2$ \cite{Hastad01}.
We give proof schemes of a result by Gutin et al. \cite{GutinKimSzeiderYeo09a} that {\sc Max-$r$-Lin2-AA} has a kernel of quadratic size and a result of Crowston, Fellows et al. \cite{CroFelGut} that {\sc Max-$r$-Lin2-AA} has a kernel with at most $(2k-1)r$ variables. The latest result improves that of Kim and Williams \cite{KimWil} that {\sc Max-$r$-Lin2-AA} has a kernel with at most $r(r+1)k$ variables. Papers  \cite{CroFelGut,KimWil} imply an algorithm of runtime $2^{O(k)}+m^{O(1)}$ for {\sc Max-$r$-Lin2-AA}.

We give a proof scheme of a result by Alon et al. \cite{AlonEtAl2009a} that {\sc Max-$r$-CSP-AA} has a a kernel of
polynomial size.  The main idea of the proof is to reduce {\sc Max-$r$-CSP-AA} to {\sc Max-$r$-Lin2-AA} and use results on {\sc Max-$r$-Lin2-AA} and a lemma
on bikernels given in the next section. The result of Alon et al. \cite{AlonEtAl2009a}
solves an open question of Mahajan, Raman and Sikdar \cite{MahajanRamanSikdar09} not only for {\sc Max-$r$-Sat-AA} but for the more general problem {\sc Max-$r$-CSP-AA}.
The problem {\sc Max-$r$-Sat-AA} is a special case of {\sc Max-$r$-CSP-AA} when
every formula is a clause with at most
$r$ variables. For {\sc Max-$r$-Sat-AA}, the reduction to {\sc Max-$r$-Lin2-AA} can be complemented by a reduction from {\sc Max-$r$-Lin2-AA}
back to {\sc Max-$r$-Sat-AA}, which yields a kernel of quadratic size. (Note that while the size of the kernel for {\sc Max-$r$-CSP-AA} is polynomial
we are unable to bound the degree of the polynomial.)

{\sc MaxLin2-AA} is the same problem as {\sc Max-$r$-Lin2-AA}, but the number of variables in an equation is not bounded. Thus, {\sc MaxLin2-AA} is a generalization of {\sc Max-$r$-Lin2-AA}. Section \ref{sec:maxlin} presents a scheme of a recent proof by Crowston, Fellows et al. \cite{CroFelGut} that {\sc MaxLin2-AA} is fixed-parameter tractable and has a kernel with polynomial number of variables. This result finally solved an open question of Mahajan, Raman and Sikdar \cite{MahajanRamanSikdar09}. Still, we do not know whether {\sc MaxLin2-AA} has a kernel of polynomial size and we present only partial results on the topic. {\sc Max-Sat-AA} is the same problem as {\sc Max-$r$-Sat-AA}, but the number of variables in a clause is not bounded. Crowston et al. \cite{CroGutJonRamSau} proved that {\sc Max-Sat-AA} is para-NP-complete and, thus, {\sc MaxSat-AA} is not fixed-parameter tractable unless P$=$NP.
We give a short discussion of this result in the end of Section \ref{sec:maxlin}.

In Section \ref{sec:perm} we discuss parameterizations above average of Ordering CSPs of arities 2 and 3. It turns out that
for our parameterization the most important Ordering CSP is the problem \textsc{$r$-Linear Ordering} ($r \geq 2$).
An instance of \textsc{$r$-Linear Ordering}
consists of a set $V$ of variables and a multiset $C$ of constraints, which are
ordered $r$-tuples of distinct variables of $V$ (note that the same set of $r$ variables
may appear in several different constraints). The objective is
to find an ordering $\alpha$ of $V$ that maximizes
the number of constraints whose order in $\alpha$
follows that of the constraint (such constraints are {\em satisfied} by $\alpha$).

It is easy to see that $|C|/r!$ is the average number of constraints satisfied by an ordering of $V$ and that it is a tight lower bound on the maximum number
of constraints satisfied by an ordering of $V$. The only nontrivial Ordering CSP of arity 2 is {\sc 2-Linear Ordering}. For this problem, Guruswami, Manokaran and Raghavendra \cite{GurManRag} proved that it is impossible to find, in polynomial time, an ordering that satisfies at least $|C|(1+\epsilon)/2$ constraints for every $\epsilon >0$ provided the Unique Games Conjecture (UGC) of Khot \cite{khot} holds.  Similar approximation resistant results were proved for all Ordering CSPs of arity 3 by  Charikar, Guruswami and Manokaran \cite{ChaGurMan}  and for Ordering CSPs of any arity by Guruswami et al. \cite{GurHasManRagCha}.

 In the problem {\sc $r$-Linear Ordering} parameterized above average ({\sc $r$-Linear Ordering-AA}), given an instance of {\sc $r$-Linear Ordering} with
a multiset $C$ of constraints, we are to decide whether there is an ordering satisfying at least $|C|/r!+k$ constraints,  where $k$ is the parameter.
Gutin et al. \cite{GutinKimSzeiderYeo09a} proved that {\sc 2-Linear Ordering-AA} is fixed-parameter tractable and, moreover, has a kernel of a quadratic size.
{\sc Betweenness} is an Ordering CSP of arity 3, which is  formulated  in Section \ref{sec:perm}.
Gutin et al. \cite{GutinKimMnichYeo} solved an open question of Benny Chor stated in Niedermeier's monograph \cite{Niedermeier06} by showing that {\sc Betweenness} parameterized above average is fixed-parameter tractable and, moreover, has a kernel of a quadratic size.

A simple, yet important, observation is that all Ordering CSPs of arity 3 parameterized above average can be reduced, in polynomial time, to {\sc 3-Linear Ordering} parameterized above average ({\sc 3-Linear Ordering-AA}) and that this reduction preserves the parameter. Thus, to prove that all Ordering CSPs of arity 3 parameterized above average are fixed-parameter tractable, it suffices to show that  {\sc 3-Linear Ordering-AA} is fixed-parameter tractable. Gutin et al. \cite{GutinIerselMnichYeo} proved that {\sc 3-Linear Ordering-AA} is fixed-parameter tractable and, moreover, has a kernel with a quadratic number of constraints and variables.

Kim and Williams \cite{KimWil} partially improved the results above by showing that {\sc 2-Linear Ordering-AA} and {\sc 3-Linear Ordering-AA} have kernels with linear number of variables. Parameterized complexity of Ordering CSPs of arities 4 and higher is still unknown. It seems to be technically much more difficult to prove that {\sc 4-Linear Ordering-AA} is fixed-parameter tractable than that {\sc 3-Linear Ordering-AA} is fixed-parameter tractable.

\section{Basics on Parameterized Algorithms and Complexity} \label{sec:tn}

A parameterized problem $\Pi$ can be considered as a set of pairs
$(I,k)$ where $I$ is the \emph{problem instance} and $k$ (usually a nonnegative
integer) is the \emph{parameter}.  $\Pi$ is called
\emph{fixed-parameter tractable (fpt)} if membership of $(I,k)$ in
$\Pi$ can be decided by an algorithm of runtime $O(f(k)|I|^c)$, where $|I|$ is the size
of $I$, $f(k)$ is an arbitrary function of the
parameter $k$ only, and $c$ is a constant
independent from $k$ and $I$. Such an algorithm is called an {\em fpt} algorithm.
Let $\Pi$ and $\Pi'$ be parameterized
problems with parameters $k$ and $k'$, respectively. An
\emph{fpt-reduction $R$ from $\Pi$ to $\Pi'$} is a many-to-one
transformation from $\Pi$ to $\Pi'$, such that (i) $(I,k)\in \Pi$ if
and only if $(I',k')\in \Pi'$ with $k'\le g(k)$ for a fixed
computable function $g$, and (ii) $R$ is of complexity
$O(f(k)|I|^c)$.

If the nonparameterized version of $\Pi$ (where $k$ is just part of the input)
is NP-hard, then the function $f(k)$ must be superpolynomial
provided P$\neq$NP. Often $f(k)$ is ``moderately exponential,''
which makes the problem practically feasible for small values of~$k$.
Thus, it is important to parameterize a problem in such a way that the
instances with small values of $k$ are of real interest.

When the decision time is replaced by the much more powerful $|I|^{O(f(k))},$
we obtain the class XP, where each problem is polynomial-time solvable
for any fixed value of $k.$ There is an infinite number of parameterized complexity
classes between FPT and XP (for each integer $t\ge 1$, there is a class W[$t$]) and they form the following tower:
$$FPT \subseteq W[1] \subseteq W[2] \subseteq \cdots \subseteq W[P] \subseteq XP.$$
Here W[P] is the class of all parameterized problems $(x,k)$ that can be decided in $f(k)|x|^{O(1)}$ time
by a nondeterministic Turing machine that makes at most $f(k)\log |I|$ nondeterministic steps for some function $f$.
For the definition of classes W[$t$],
see, e.g., \cite{FlumGrohe06} (we do not use these classes in the rest of the paper).

$\Pi$ is in \emph{para-NP} if membership of $(I,k)$ in
$\Pi$ can be decided in nondeterministic time $O(f(k)|I|^c)$,
where $|I|$ is the size of $I$, $f(k)$ is an arbitrary function of the
parameter $k$ only,
and $c$ is a constant independent from $k$ and $I$. Here,
nondeterministic time means that we can use nondeterministic
Turing machine. A parameterized problem $\Pi'$ is {\em
para-NP-complete} if it is in para-NP and for any parameterized
problem $\Pi$ in para-NP there is an fpt-reduction from $\Pi$ to
$\Pi'$.

While several fpt algorithms were designed many years ago (e.g., pseudo-polynomial algorithms with parameter
being the binary length of the maximum number, cf. \cite{GareyJohnson79}),
Downey and Fellows were the first to systematically study the theory of parameterized algorithms and complexity
and they wrote the first monograph \cite{DowneyFellows99} in the area\footnote{Michael R. Fellows has worked tirelessly for many years to promote the area and so can be affectionately called St. Paul of Parameterized Complexity.}.

Given a pair $\Pi,\Pi'$ of parameterized problems,
a \emph{bikernelization from $\Pi$ to $\Pi'$} is a polynomial-time
algorithm that maps an instance $(I,k)$ to an instance $(I',k')$ (the
\emph{bikernel}) such that (i)~$(I,k)\in \Pi$ if and only if
$(I',k')\in \Pi'$, (ii)~ $k'\leq f(k)$, and (iii)~$|I'|\leq g(k)$ for some
functions $f$ and $g$. The function $g(k)$ is called the {\em size} of the bikernel.
A {\em kernelization} of a parameterized problem
$\Pi$ is simply a bikernelization from $\Pi$ to itself. Then $(I',k')$ is a {\em kernel}.
The term bikernel was coined by Alon et al. \cite{AlonEtAl2009a}; in \cite{BodlaenderEtAl2009a}
a bikernel is called a generalized kernel.

It is well-known that
a parameterized problem $\Pi$ is fixed-parameter
tractable if and only if it is decidable and admits a
kernelization \cite{DowneyFellows99,FlumGrohe06,Niedermeier06}. This result can be extended as follows:
A decidable parameterized problem $\Pi$ is fixed-parameter
tractable if and only if it admits a
bikernelization from itself to a decidable parameterized problem $\Pi'$ \cite{AlonEtAl2009a}.

Due to applications, low degree polynomial size kernels are of main interest. Unfortunately,
many fixed-parameter tractable problems do not have
kernels of polynomial size unless the polynomial hierarchy collapses to the third level \cite{BodlaenderEtAl2009a,BodlaenderEtAl2009,Fernau2009}.
For further background and terminology on parameterized complexity we
refer the reader to the monographs~\cite{DowneyFellows99,FlumGrohe06,Niedermeier06}.

The following lemma of Alon et al. \cite{AlonEtAl2009a} inspired by a lemma from \cite{BodlaenderEtAl2009}
shows that polynomial bikernels imply polynomial kernels.

\begin{lemma}\label{lem:pk}
Let $\Pi,\Pi'$ be a pair of decidable parameterized problems such that the nonparameterized version of $\Pi'$
is in NP, and the nonparameterized version of $\Pi$ is NP-complete.
If there is a bikernelization from $\Pi$ to $\Pi'$ producing a
bikernel of polynomial size, then $\Pi$ has a polynomial-size kernel.
\end{lemma}

Henceforth $[n]$ stands for the set $\{1,2,\ldots , n\}.$

\section{Probabilistic and Harmonic Analysis Tools}\label{sec:prob}

We start this section by outlining the very basic principles of the probabilistic method which will be
implicitly used in this paper. Given random variables $X_1, \ldots , X_n$, the fundamental property known as
{\em linearity of expectation} states that $\Exp(X_1+\ldots +X_n)=\Exp(X_1)+ \ldots + \Exp(X_n)$.
The {\em averaging argument} utilizes the fact that there is a point for which $X\geq \Exp(X)$ and a
point for which $X\leq \Exp(X)$ in the probability space. Also a positive probability $\Prob(A)>0$ for
some event $A$ means that there is at least one point in the probability space which belongs to $A$.
For example, $\Prob(X\geq k)>0$ tells us that there exists a point for which $X\geq k$.

A random variable is {\em discrete} if its distribution function
has a finite or countable number of positive increases.
A random variable $X$ is {\em symmetric}  if $-X$ has the same distribution function as $X$.
If $X$ is discrete, then $X$ is symmetric if and only if $\Prob(X=a)=\Prob(X=-a)$
for each real $a.$  Let $X$ be a symmetric variable for which the first moment $\mathbb{E}(X)$ exists.
Then $\mathbb{E}(X)=\mathbb{E}(-X)=-\mathbb{E}(X)$ and, thus, $\mathbb{E}(X)=0.$
The following is easy to prove \cite{GutinKimSzeiderYeo09a}.
\begin{lemma}\label{eqsym}
If $X$ is a symmetric random variable and $\mathbb{E}(X^2)$ is finite, then
$$\Prob(\ X \ge \sqrt{\mathbb{E}(X^2)}\ )>0.$$\end{lemma}

If $X$ is not symmetric then the following lemma can be used instead (a similar result was
already proved in \cite{AGK04}).
\begin{lemma}[Alon et al.~\cite{AlonEtAl2009a}]
\label{lem32}
  Let $X$ be a real random variable and suppose that its first, second and fourth moments satisfy $\mathbb{E}[X] = 0$, $\mathbb{E}[X^2] = \sigma^2 > 0$ and $\mathbb{E}[X^4] \leq c \mathbb{E}[X^2]^2$, respectively, for some constant $c$.
  Then $\Prob(X > \frac{\sigma}{2 \sqrt c}) > 0$.
\end{lemma}

\noindent
To check $\mathbb{E}[X^4] \leq c \mathbb{E}[X^2]^2$ we often can use the following well-known inequality.

\begin{lemma}[Hypercontractive Inequality~\cite{Bonami1970}]
\label{lem41}
  Let $f = f(x_1,\ldots,x_n)$ be a polynomial of degree $r$ in $n$ variables $x_1,\ldots,x_n$ each with domain $\{-1,1\}$.
  Define a random variable $X$ by choosing a vector $(\epsilon_1,\ldots,\epsilon_n)\in \{-1,1\}^n$ uniformly at random and setting $X = f(\epsilon_1,\ldots,\epsilon_n)$.
  Then $\mathbb E[X^4]\leq 9^r\mathbb E[X^2]^2$.
\end{lemma}

If $f = f(x_1,\ldots,x_n)$ is a polynomial in $n$ variables $x_1,\ldots,x_n$ each with domain $\{-1,1\}$,
then it can be written as $f=\sum_{I\subseteq [n]} c_I \prod_{i\in S}x_i$, where $[n]=\{1,\ldots ,n\}$ and $c_I$ is a real for each $I\subseteq [n].$

The following dual, in a sense, form of the Hypercontractive Inequality was proved by Gutin and Yeo \cite{GutinYeo2011}; for a weaker result, see \cite{GutinKimSzeiderYeo09a}.

\begin{lemma}\label{lem:HI2}
Let $f = f(x_1,\ldots,x_n)$ be a polynomial in $n$ variables $x_1,\ldots,x_n$ each with domain $\{-1,1\}$ such that $f=\sum_{I\subseteq [n]} c_I \prod_{i\in S}x_i$.
Suppose that no variable $x_i$ appears in more than $\rho$ monomials of $f$.
Define a random variable $X$ by choosing a vector $(\epsilon_1,\ldots,\epsilon_n)\in \{-1,1\}^n$ uniformly at random and setting $X = f(\epsilon_1,\ldots,\epsilon_n)$.
Then $\mathbb E[X^4]\leq (2\rho+1)\mathbb E[X^2]^2$.
\end{lemma}

The following lemma is easy to prove, cf. \cite{GutinKimSzeiderYeo09a}. In fact, the equality there is a special case of Parseval's Identity in Harmonic Analysis, cf. \cite{odonn}.

\begin{lemma}\label{lem:Pars}
Let $f = f(x_1,\ldots,x_n)$ be a polynomial in $n$ variables $x_1,\ldots,x_n$ each with domain $\{-1,1\}$ such that $f=\sum_{I\subseteq [n]} c_I \prod_{i\in I}x_i$. Define a random variable $X$ by choosing a vector $(\epsilon_1,\ldots,\epsilon_n)\in \{-1,1\}^n$ uniformly at random and setting $X = f(\epsilon_1,\ldots,\epsilon_n)$.
Then $\mathbb E[X^2]=\sum_{i\in I}c^2_I$.
\end{lemma}

\section{Parameterizations of  MaxSat}\label{sec:sat}

In the well-known problem {\sc MaxSat}, we are given a CNF formula $F$ with $m$ clauses and asked to
determine the maximum number of clauses of $F$ that can be satisfied by a truth assignment.
Let us assign {\sc True} to each variable of $F$ with probability $1/2$ and observe that the probability of a clause
to be satisfied is at least $1/2$ and thus, by linearity of expectation, the expected number of satisfied clauses in $F$
is at least $m/2$. Thus, by the averaging argument, there exists a truth assignment to the variables of $F$ which satisfies
at least $m/2$ clauses of $F$.

Let us denote by ${\rm sat}(F)$ the maximum number of clauses of $F$ that can be satisfied by a truth assignment.
The lower bound ${\rm sat}(F)\ge m/2$ is tight as we have ${\rm sat}(H)=m/2$ if $H=(x_1)\wedge(\bar{x}_1)\wedge \cdots \wedge (x_{m/2})\wedge(\bar{x}_{m/2})$.
Consider the following parameterization of {\sc MaxSat} above tight lower bound introduced by Mahajan and Raman \cite{Mahajan97}.

\begin{quote}
{\sc MaxSat-A}($m/2$)\\ \nopagebreak
  \emph{Instance:} A CNF formula $F$ with $m$ clauses (clauses may appear several times in $F$) and a nonnegative integer $k$.\\
    \nopagebreak
  \emph{Parameter:} $k$.\\ \nopagebreak
  \emph{Question:} ${\rm sat}(F)\ge m/2 + k$?
 \end{quote}

Mahajan and Raman  \cite{Mahajan97} proved that {\sc MaxSat-A}($m/2$) admits a kernel with at most $6k+3$ variables and $10k$ clauses.
Crowston et al. \cite{CrowstonIPEC} improved this result, by obtaining a kernel with at most $4k$ variables and $(2\sqrt{5}+4)k$ clauses.
The improved result is a simple corollary of a new lower bound on ${\rm sat}(F)$ obtained in  \cite{CrowstonIPEC},
which is significantly stronger than the simple bound ${\rm sat}(F)\ge m/2$. We give the new lower bound below, in Theorem \ref{thm:CGJYbound}.

For a variable $x$ in $F$, let $m(x)$ denote
the number of pairs of unit of clauses $(x),(\bar{x})$ that have to be deleted from $F$ such that $F$ has no pair $(x),(\bar{x})$ any longer.
Let ${\rm var}(F)$ be the set of all variables in $F$ and let $\ddot{m}=\sum_{x\in {\rm var}(F)}m(x).$ The following is a stronger lower bound on
${\rm sat}(F)$ than $m/2$.

\begin{theorem}\label{thm:LSbound}
For a CNF formula $F$, we have ${\rm sat}(F)\ge \ddot{m}/2 +  \pp (m-\ddot{m})$, where $\pp =(\sqrt{5}-1)/2\approx 0.618$.
\end{theorem}

A CNF formula $F$ is \emph{$t$-satisfiable} if for any $t$ clauses in $F$,
there is a truth assignment which satisfies all of them. It is easy to check that $F$ is 2-satisfiable if and only if
$\ddot{m}=0$ and clearly Theorem \ref{thm:LSbound} is equivalent to the assertion that if $F$ is 2-satisfiable then ${\rm sat}(F)\ge \pp m$.
The proof of this assertion by Lieberherr and Specker \cite{LieberherrSpecker81} is quite long; Yannakakis \cite{Yannakakis94} gave the following short probabilistic
proof. For $x\in {\rm var}(F)$, let the probability of $x$ being assigned {\sc True} be $\pp$ if $(x)$ is in $F$, $1-\pp$ if
$(\bar{x})$ is in $F$, and $1/2$, otherwise, independently of the other variables.
Let us bound the probability $p(C)$ of a clause $C$ to be satisfied. If $C$ contains only one literal,
then, by the assignment above,  $p(C)=\pp.$ If $C$ contains two literals, then, without loss of generality, $C=(x \vee y)$. Observe that
the probability of $x$ assigned {\sc False} is at most $\pp$ (it is $\pp$ if $(\bar{x})$ is in $F$). Thus, $p(C)\ge 1-\pp^2.$ It remains to observe that
$ 1-\pp^2=\pp.$ Now to obtain the bound ${\rm sat}(F)\ge \pp m$ apply linearity of expectation and the averaging argument.

Note that $\pp m$ is an {\em asymptotically} tight lower bound: for each $\epsilon >0$ there are 2-satisfiable CNF formulae $F$ with
${\rm sat}(F) < m(\pp + \epsilon)$ \cite{LieberherrSpecker81}. Thus, the following problem stated by Mahajan and Raman \cite{Mahajan97} is natural.

\begin{quote}
{\sc Max-2S-Sat-A}($\pp m$)\\ \nopagebreak
  \emph{Instance:} A 2-satisfiable CNF formula $F$ with $m$ clauses (clauses may appear several times in $F$) and a nonnegative integer $k$.\\
    \nopagebreak
  \emph{Parameter:} $k$.\\ \nopagebreak
  \emph{Question:} ${\rm sat}(F)\ge \pp m + k$?
 \end{quote}

Mahajan and Raman \cite{Mahajan97} conjectured that {\sc Max-2S-Sat-A}($\pp m$) is fpt.
Crowston et al. \cite{CrowstonIPEC} solved this conjecture in the affirmative; moreover, they obtained a kernel with at most $(7+3\sqrt{5})k$ variables. This result is an easy corollary from a lower bound on ${\rm sat}(F)$ given in Theorem \ref{thm:CGJYbound}, which, for 2-satisfiable CNF formulas, is stronger than the one in Theorem \ref{thm:LSbound}. The main idea of \cite{CrowstonIPEC} is to obtain a lower bound on ${\rm sat}(F)$ that includes the number of variables as a factor. It is clear that for general CNF formula $F$ such a bound is impossible. For consider a formula containing a single clause $c$ containing a large number of variables. We can arbitrarily increase the number of variables in the formula, and the maximum number of satisfiable clauses will always be 1.
We therefore need a reduction rule that cuts out `excess' variables. Our reduction rule is based on the notion of an expanding formula given below.
Lemma \ref{lem:expformula} and Theorem \ref{thm:expth} show the usefulness of this notion.

A CNF formula $F$ is called {\em expanding} if for each $X\subseteq {\rm var}(F)$, the number of clauses containing at least one variable from $X$ is at least $|X|$ \cite{FKS2002,Sze2004}. The following lemma and its parts were proved by many authors, see, e.g., Fleischner, Kullmann and Szeider \cite{FKS2002}, Lokshtanov \cite{LokshtanovPhD} and Szeider \cite{Sze2004}.

\begin{lemma}\label{lem:expformula}
Let $F$ be a CNF formula and let $V$ and $C$ be its sets of variables and clauses.
There exists a subset $C^*\subseteq C$ that can
be found in polynomial time, such that the formula $F'$ with clauses $C\setminus C^*$ and
variables $V\setminus V^*$, where $V^*={\rm var}(C^*)$, is expanding. Moreover, ${\rm sat}(F)={\rm sat}(F')+|C^*|.$
\end{lemma}

The following result was shown by Crowston et al. \cite{CrowstonIPEC}. The proof is nontrivial and consists of a deterministic algorithm for finding the corresponding truth assignment and a detailed combinatorial analysis of the algorithm.

\begin{theorem}\label{thm:expth}
Let $F$ be an expending 2-satisfiable CNF formula with $n$ variables and $m$ clauses.
Then ${\rm sat}(F)\ge \pp m + n(2-3\pp )/2.$
\end{theorem}

Lemma \ref{lem:expformula} and Theorem \ref{thm:expth} imply the following:

\begin{theorem}\label{thm:CGJYbound}
Let $F$ be a 2-satisfiable CNF formula and let $V$ and $C$ be its sets of variables and clauses.
There exists a subset $C^*\subseteq C$ that can
be found in polynomial time, such that the formula $F'$ with clauses $C\setminus C^*$ and variables $V\setminus V^*$, where $V^*={\rm var}(C^*)$, is expanding. Moreover, we have $${\rm sat}(F)\ge \pp m + (1-\pp)m^* + (n-n^*)(2-3\pp )/2,$$ where $m=|C|,$ $m^*=|C^*|$, $n=|V|$ and $n^*=|V^*|.$
\end{theorem}

Let us turn now to 3-satisfiable CNF formulas. If $F$ is $3$-satisfiable then it is not hard to check that the forbidden sets of clauses are pairs of
the form $\{x\}, \{\bar{x}\}$ and triplets of the form $\{x\}, \{y\},
\{\bar{x}, \bar{y}\}$ or $\{x\}, \{\bar{x}, y\}, \{\bar{x}, \bar{y}\}$, as
well as any triplets that can be derived from these by switching positive
literals with negative literals.

Lieberherr and Specker \cite{LieberherrSpecker82} and, later, Yannakakis  \cite{Yannakakis94} proved the following:
if $F$ is $3$-satisfiable then ${\rm sat}(F) \ge \frac{2}{3}w({\cal C}(F))$. This bound is also asymptotically tight.
Yannakakis \cite{Yannakakis94} gave a probabilistic proof which is similar to his proof for 2-satisfiable formulas, but requires consideration
of several cases and, thus, not as short as for 2-satisfiable formulas. For details of his proof, see, e.g., Gutin, Jones and Yeo \cite{GutJonYeoFCT11}
and Jukna \cite{Juk2001} (Theorem 20.6). Yannakakis's approach was extended by Gutin, Jones and Yeo \cite{GutJonYeoFCT11} to prove the following theorem using a quite complicated probabilistic distribution for a random truth assignment.

\begin{theorem}\label{thm:3s-sat}
Let $F$ be an expanding 3-satisfiable CNF formula with $n$ variables and $m$ clauses. Then ${\rm sat}(F)\ge \frac{2}{3}m + \rho n,$
where  $\rho(>0.0019)$ is a constant.
\end{theorem}

This theorem and Lemma \ref{lem:expformula} imply the following:

\begin{theorem}\label{thm:GJYbound}
Let $F$ be a 3-satisfiable CNF formula and let $V$ and $C$ be its sets of variables and clauses.
There exists a subset $C^*\subseteq C$ that can
be found in polynomial time, such that the formula $F'$ with clauses $C\setminus C^*$ and variables $V\setminus V^*$, where $V^*={\rm var}(C^*)$, is expanding. Moreover, we have $${\rm sat}(F)\ge \frac{2}{3} m + \frac{1}{3}m^* + \rho(n-n^*),$$ where $\rho(>0.0019)$ is a constant, $m=|C|,$ $m^*=|C^*|$, $n=|V|$ and $n^*=|V^*|.$
\end{theorem}

Using this theorem it is easy to obtain a linear-in-number-of-variables kernel for the following natural analog of {\sc Max-2S-Sat-A}($\pp m$), see \cite{GutJonYeoFCT11} for details.

\begin{quote}
{\sc Max-3S-Sat-A}($\frac{2}{3}m$)\\ \nopagebreak
  \emph{Instance:} A 3-satisfiable CNF formula $F$ with $m$ clauses and a nonnegative integer $k$.\\
    \nopagebreak
  \emph{Parameter:} $k$.\\ \nopagebreak
  \emph{Question:} ${\rm sat}(F)\ge  \frac{2}{3} m + k$?
 \end{quote}

Now let us consider the following important parameterization of {\sc $r$-Sat} below the tight upper bound $m$:
\begin{quote}
{\sc $r$-Sat-B}($m$)\\ \nopagebreak
  \emph{Instance:} An $r$-CNF formula $F$ with $m$ clauses (every clause has at most $r$ literals) and a nonnegative integer $k$.\\
    \nopagebreak
  \emph{Parameter:} $k$.\\ \nopagebreak
  \emph{Question:} ${\rm sat}(F)\ge m - k$?
 \end{quote}
Since {\sc Max-$r$-Sat} is NP-hard for each fixed $r\ge 3$, {\sc $r$-Sat-B}($m$) is not fpt unless
P$=$NP. However, the situation changes for $r=2$:  Razgon and O'Sullivan \cite{RazOsu} proved that {\sc 2-Sat-B}($m$) is fpt. The algorithm in \cite{RazOsu} is of complexity $O(15^kk m^3)$ and, thus, {\sc Max-2-Sat-B}($m$) admits a kernel with at most $15^kk$ clauses.
It is not known whether {\sc 2-Sat-B}($m$) admits a kernel with a polynomial number of variables. Raman et al. \cite{RamRamSau} and Cygan et al. \cite{CygPilPilWoj} designed algorithms for {\sc 2-Sat-B}($m$) of runtime $9^k(km)^{O(1)}$ and $4^k(km)^{O(1)}$, respectively.
In both papers, the authors consider the following parameterized problem ({\sc VC-AMM}): given a graph $G$ whose maximum matching is of cardinality $\mu$, decide whether $G$ has a vertex cover with at most $\mu + k$ vertices, where $k$ is the parameter. A parameterized algorithm of the above-mentioned complexity actually is obtained for {\sc VC-AMM}, and {\sc 2-Sat-B}($m$) is polynomially transformed into {\sc VC-AMM} (the transformation is parameter-preserving).
While Raman et al. \cite{RamRamSau} obtain the parameterized algorithm for {\sc VC-AMM} directly, Cygan et al. \cite{CygPilPilWoj} derive it via a reduction from a more general problem on graphs parameterized above a tight bound.

{\sc 2-Sat-B}($m$) has several application. {\sc 2-Sat-B}($m$) is, in fact, equivalent to {\sc VC-AMM} \cite{MisALGO,RamRamSau,CygPilPilWoj}.  Mishra et al. \cite{MisALGO} studied the following problem: given a graph $G$, decide whether by deleting at most $k$ vertices we can make $G$ {\em K{\"o}nig}, i.e., a graph in which the minimum size of a vertex cover equals the maximum number of edges in a matching. They showed how to reduce the last problem to {\sc VC-AMM}.
It is noted by  Gottlob and Szeider \cite{GotSze} that fixed-parameter tractability of {\sc VC-AMM} implies the fixed-parameter tractability
of the following problem. Given a CNF formula $F$ (not necessarily 2-CNF), decide whether there exists a subset of at most $k$ variables of $F$
so that after removing all occurrences of these variables from the clauses of $F$, the resulting CNF formula is {\em Renamable Horn}, i.e., it can
be transformed by renaming of the variables into a CNF formula with at most one positive literal in each clause.

{\sc 2-SAT-B}$(m)$ has also been used in order to obtain the best known bound on the order of a kernel for {\sc Vertex Cover} (given a graph $G$ and an integer $k$, decide whether $G$ has a vertex cover with at most $k$ vertices). The fact that {\sc Vertex Cover} has a kernel with at most $2k$ vertices was known for a long time, see Chen,  Kanj and Jia \cite{Chen}. This was improved to $2k-1$ by Chleb\'{i}k and Cleb\'{i}kov\'{a} \cite{Chlebik} and further to $2k-c$ for any constant $c$ by Soleimanfallah and Yeo \cite{ASYeo11}. Lampis \cite{Lampis} used the same approach as in \cite{ASYeo11}, but instead of reducing an instance of {\sc Vertex Cover} to a large number of {\sc 2-SAT} instances, he reduced {\sc Vertex Cover} to {\sc 2-SAT-B}$(m)$ via {\sc VC-AMM}. As a result, Lampis \cite{Lampis} obtained a kernel of order at most $2k- c \log k$  for any constant $c$. We will now briefly describe how this kernel was obtained.

For a graph $G$ let $\beta(G)$ denotes the minimum size of a vertex cover of $G$ and $\mu(G)$ the maximum size of a matching in $G$. In their classical work Nemhauser and Trotter \cite{NemTro} proved the following:

\begin{theorem}\label{thm:NT} 
There is an $O(|E|\sqrt{|V|})$-time algorithm which for a given graph $G=(V,E)$ computes two
disjoint subsets of vertices of $G,$ $V'$, $V''$, such that $\beta(G) = \beta(G[V']) + |V''|$ and $\beta(G[V'])\ge |V'|/2.$
\end{theorem}

Soleimanfallah and Yeo \cite{ASYeo11} showed the following additional inequality: 
\begin{equation}\label{eq:SY} \beta(G[V']) \ge |V'| - \mu(G).\end{equation}

Let $k':=k-|V''|.$ By Theorem \ref{thm:NT}, $\beta(G)\le k$ if and only if $\beta(G[V'])\le k'$. If $|V'| \le  2k'-c \log k'\le  2k -c \log k$
then we have a kernel and we are done. Thus, it suffices to show that if $|V'| > 2k'-c \log k'$ we can decide whether $ \beta(G[V'])\le k'$ in polynomial time.
We assume that $|V'| > 2k'-c \log k'$ and we may also assume that $|V'|\le 2k'$ as otherwise $ \beta(G[V'])> k'$ by Theorem \ref{thm:NT}. By (\ref{eq:SY}) if
$\mu(G[V'])\le (|V'|-c \log k')/2$ then $\beta(G[V']) \ge (|V'|+c \log k')/2.$ Since $|V'| > 2k' -  c \log k'$ this
means that $\beta(G[V']) > k'$. 

So, consider the case $\mu(G[V']) > (|V'|-c \log k')/2$. Since $|V'| > 2k'-c \log k'$ and $\mu(G[V']) > (|V'|-c \log k')/2,$
we have $\mu(G[V']) > k'- c \log k'$ and so $k'< \mu(G[V']) + c \log k'$. Thus, to decide whether $ \beta(G[V'])\le k'$ it suffices to compute $\ell$ such that $ \beta(G[V'])= \mu(G[V'])+\ell$, where $\ell < c \log k'$, and to compare $\mu(G[V'])+\ell$ with $k'$. Using an fpt algorithm for {\sc VC-AMM} (which is essentially an fpt algorithm for {\sc Max-2-Sat-B}($m$) as the two problems are equivalent) we can compute $\ell$ in fpt time (provided we use an efficient algorithm such as in \cite{RazOsu,RamRamSau,CygPilPilWoj}).

\section{Boolean Max-$r$-CSPs Above Average}\label{sec:csp}

Throughout this section, $r$ is a positive integral constant.
Recall that the problem {\sc Max-$r$-CSP-AA} is given by a set $V$ of $n$ variables
and a set of $m$ Boolean formulas; each formula is assigned an integral positive weight and contains at most $r$ variables from $V$.
Averaging over all truth assignments, we can find the average value $A$ of the weight of satisfied formulas. We wish to
decide whether there exists a truth assignment that satisfies formulas of total weight at least $A+k,$ where $k$ is the parameter ($k$ is a nonnegative integer).

Recall that the problem {\sc Max-$r$-Lin2-AA} is a special case of {\sc Max-$r$-CSP-AA} when every formula is a linear equation
over $\mathbb{F}_2$ with at most $r$ variables and
that {\sc Max-Lin2-AA} is the extension of {\sc Max-$r$-Lin2-AA} when we do not bound the number of variables in an equation.
Research of both {\sc Max-$r$-Lin2-AA} and {\sc Max-Lin2-AA} led to a number of basic notions and results of interest for both problems, and we devote Subsection \ref{sec:basicr} to these notions and results. In particular, we will show that $A=W/2$, where $W$ is the total weight of all equations,
introduce a Gaussian-elimination-type algorithm for both problems, and a notion and simple lemma of a sum-free subset of a set of vectors in $\mathbb{F}^n_2$. This lemma is a key ingredient in proving  some important results for {\sc Max-$r$-Lin2-AA} and {\sc Max-Lin2-AA}.

{\sc Max-$r$-Lin2-AA} is studied in Subsection \ref{sec:maxrlin}, where we give proof schemes of a result by Gutin et al. \cite{GutinKimSzeiderYeo09a} that {\sc Max-$r$-Lin2-AA} has a kernel of quadratic size and a result of Crowston, Fellows et al. \cite{CroFelGut} that {\sc Max-$r$-Lin2-AA} has a kernel with at most $(2k-1)r$ variables. The latest result improves that of Kim and Williams \cite{KimWil} that {\sc Max-$r$-Lin2-AA} has a kernel with at most $r(r+1)k$ variables.

In Subsection \ref{subsec:max-r-csp}, we give a proof scheme of a result by Alon et al. \cite{AlonEtAl2009a} that {\sc Max-$r$-CSP-AA} has a a kernel of polynomial size.  The main idea of the proof is to reduce {\sc Max-$r$-CSP-AA} to {\sc Max-$r$-Lin2-AA} and use the above results on {\sc Max-$r$-Lin2-AA}
and Lemma \ref{lem:pk}. This shows the existence of a polynomial-size kernel, but does not allow us to obtain a bound on the degree of the polynomial.
Nevertheless, this solves an open question of Mahajan, Raman and Sikdar \cite{MahajanRamanSikdar09} not only for {\sc Max-$r$-Sat-AA} but also for the more general problem {\sc Max-$r$-CSP-AA}. Recall that the problem {\sc Max-$r$-Sat-AA} is a special case of {\sc Max-$r$-CSP-AA} when every formula is a
clause with at most $r$ variables. For {\sc Max-$r$-Sat-AA}, the reduction to {\sc Max-$r$-Lin2-AA} can be complemented by a reduction from {\sc Max-$r$-Lin2-AA} back to {\sc Max-$r$-Sat-AA}, which yields a kernel of quadratic size.

\subsection{Basic Results for Max-Lin2-AA and Max-$r$-Lin2-AA} \label{sec:basicr}

Recall that in the problems \textsc{MaxLin2-AA} and \textsc{Max-$r$-Lin2-AA}, we are given a system $S$ consisting of
$m$ linear equations in $n$ variables
over $\mathbb{F}_2$ in which each equation is assigned a positive integral weight.
In \textsc{Max-$r$-Lin2-AA}, we have an extra constraint that every equation has at most $r$ variables.
Let us write the system $S$ as $\sum_{i\in I}z_i=b_I$, $I\in \cal F$, and let $w_I$ denote the weight
of an equation $\sum_{i\in I}z_i=b_I$. Clearly, $m=|{\cal F}|.$ Let $W=\sum_{I\in \cal F}w_I$
and let ${\rm sat}(S)$ be the maximum total weight of equations that can be satisfied simultaneously.

For each $i\in [n],$ set $z_i=1$ with probability 1/2 independently of the rest of the variables. Then
each equation is satisfied with probability 1/2 and the expected weight
of satisfied equations is $W/2$ (as our probability distribution is uniform, $W/2$ is also the average weight of satisfied equations).
Hence $W/2$ is a lower bound; to see
its tightness consider a system of pairs of equations of the
form $\sum_{i\in I}z_i=0,\ \sum_{i\in I}z_i=1$ of weight 1.
The aim in both \textsc{Max-Lin2-AA} and \textsc{Max-$r$-Lin2-AA}
is to decide whether for the given system $S$, ${\rm sat}(S)\ge W/2 +k,$ where $k$ is the parameter.
It is well-known that, in polynomial time, we can find an assignment to the variables that satisfies equations of total weight at least $W/2$, but, for any $\epsilon > 0$ it is NP-hard to decide whether there is an assignment satisfying equations of total weight at least $W(1+\epsilon)/2$ \cite{Hastad01}.

Henceforth, it will often be convenient for us to consider
linear equations in their multiplicative form, i.e., instead of an equation
$\sum_{i\in I}z_i=b_I$ with $z_i\in \{0,1\}$, we will consider the equation $\prod_{i\in I}x_i=(-1)^{b_I}$
with $x_i\in \{-1,1\}$. Clearly, an assignment $z^0=(z^0_1,\ldots ,z^0_n)$ satisfies $\sum_{i\in I}z_i=b_I$ if and only if
 the assignment $x^0=(x^0_1,\ldots ,x^0_n)$ satisfies $\prod_{i\in I}x_i=(-1)^{b_I},$ where $x^0_i=(-1)^{z^0_i}$ for
each $i\in [n].$

Let $\varepsilon(x)=\sum_{I\in \cal F}w_I(-1)^{b_I}\prod_{i\in I}x_i$ (each $x_i\in \{-1,1\}$)
and note that $\varepsilon(x^0)$ is the difference between the total weight of satisfied
and falsified equations when $x_i=x^0_i$ for each $i\in [n].$ Crowston et al. \cite{CrowstonSWAT} call $\varepsilon(x)$  the {\em excess} and
the maximum possible value of $\varepsilon(x)$ the {\em maximum excess}.

\begin{remark} \label{MaxLinExcess} Observe that the answer to \textsc{Max-Lin2-AA} and \textsc{Max-$r$-Lin2-AA} is {\sc Yes} if and only if
the maximum excess is at least $2k$. \end{remark}

Let $A$ be the matrix over $\mathbb{F}_2$ corresponding to the set of equations in $S$,
such that $a_{ji} = 1$ if $i \in I_j$ and $0$, otherwise.

Consider two reduction rules for {\sc Max-Lin2-AA} introduced by Gutin et al. \cite{GutinKimSzeiderYeo09a}. Rule \ref{rule1} was studied before in \cite{HasVen04}.

  \begin{krule}\label{rule1}
  If we have, for a subset $I$ of $[n]$, an equation $\prod_{i \in I} x_i =b_I'$
  with weight $w_I'$, and an equation $\prod_{i \in I} x_i =b_I''$ with weight $w_I''$,
  then we replace this pair by one of these equations with weight $w_I'+w_I''$ if $b_I'=b_I''$ and, otherwise, by
  the equation whose weight is bigger, modifying its
  new weight to be the difference of the two old ones. If the resulting weight
  is~0, we delete the equation from the system.
  \end{krule}
    \begin{krule}\label{rulerank}
  Let $t={\rm rank} A$ and suppose columns $a^{i_1},\ldots ,a^{i_t}$ of $A$ are linearly independent.
  Then delete all variables not in $\{x_{i_1},\ldots ,x_{i_t}\}$ from the equations of $S$.
  \end{krule}

  \begin{lemma}\label{lem:SS'}\cite{GutinKimSzeiderYeo09a}
  Let $S'$ be obtained from $S$ by Rule~\ref{rule1} or \ref{rulerank}.
  Then the maximum excess of $S'$  is equal to the maximum excess of $S$.
  Moreover, $S'$ can be obtained from $S$ in time polynomial in $n$ and $m$.
  \end{lemma}

If we cannot change a weighted system $S$ using  Rules~\ref{rule1} and \ref{rulerank}, we call it {\em irreducible}.

Let $S$ be an irreducible system of {\sc Max-Lin2-AA}.
Consider the following algorithm introduced in \cite{CrowstonSWAT}.
We assume that, in the beginning, no equation or variable in $S$ is marked.

\begin{center}
\fbox{~\begin{minipage}{12cm}
\textsc{Algorithm $\cal H$}

\smallskip
While the system $S$ is nonempty do the following:

1. Choose an equation  $\prod_{i \in I} x_i =b$ and
mark a variable $x_l$ such that $l \in I$.

\smallskip

2. Mark this equation and delete it from the system.

\smallskip

3. Replace every equation $\prod_{i \in I'} x_i =b'$ in the system containing $x_l$ by $\prod_{i \in I\Delta I'} x_i
= bb'$, where $I\Delta I'$ is the symmetric difference of $I$ and $I'$ (the weight of the equation is unchanged).

\smallskip

4. Apply Reduction Rule \ref{rule1} to the system.

\smallskip
\end{minipage}~}
\end{center}

\smallskip

The {\em maximum ${\cal H}$-excess} of $S$ is the maximum possible total weight of equations marked by ${\cal H}$ for $S$
taken over all possible choices in Step 1 of $\cal H$. The following lemma indicates the potential power of $\cal H$.

\begin{lemma}\label{lemExcess}\cite{CrowstonSWAT}
Let $S$ be an irreducible system. Then the maximum excess of $S$ equals its maximum ${\cal H}$-excess.
\end{lemma}

This lemma gives no indication on how to choose equations in Step 1 of Algorithm $\cal H$. As the problem {\sc Max-Lin2-AA}
is NP-hard, we cannot hope to obtain an polynomial-time procedure for optimal choice of equations in Step 1 and, thus,
have to settle for a good heuristic. For the heuristic we need the following notion first used in \cite{CrowstonSWAT}.
Let $K$ and $M$ be sets of vectors in $\mathbb{F}^n_2$ such that $K \subseteq M$. We say $K$ is \emph{$M$-sum-free}
if no sum of two or more distinct vectors in $K$ is equal to a vector in $M$. Observe that $K$ is $M$-sum-free if
and only if $K$ is linearly independent and no sum of vectors in $K$ is equal to a vector in $M \backslash K$.

The following lemma was proved implicitly in \cite{CrowstonSWAT} and, thus, we provide a short proof of this result.

\begin{lemma}\label{lem:M-freeapplic}
Let $S$ be an irreducible system of {\sc Max-Lin2-AA} and let $A$ be the matrix corresponding to $S$.
Let $M$ be the set of rows of $A$
(viewed as vectors in $\mathbb{F}^n_2$) and let $K$ be an $M$-sum-free set of $k$ vectors. Let $w_{\rm min}$ be the minimum
weight of an equation in $S$. Then, in time in $(nm)^{O(1)}$, we can find an assignment to the variables of $S$
that achieves excess of at least $w_{\rm min}\cdot k.$
\end{lemma}
\begin{proof}

Let $\{e_{j_1}, \ldots, e_{j_k}\}$ be the set of equations corresponding to the vectors in $K$.
Run Algorithm $\cal H$, choosing at Step 1
an equation of $S$ from $\{e_{j_1}, \ldots, e_{j_k}\}$ each time,
and let $S'$ be the resulting system.
Algorithm $\cal H$ will run for $k$ iterations of the while loop as
no equation from $\{e_{j_1}, \ldots, e_{j_k}\}$ will be deleted before it has been marked.

Indeed, suppose that this is not true. Then for some $e_{j_l}$ and some other equation $e$ in $S$, after
applying Algorithm $\cal H$ for at most $l-1$ iterations $e_{j_l}$ and $e$ contain the same variables.
Thus, there are vectors $v_j\in K$ and $v\in M$ and a pair of nonintersecting subsets $K'$ and $K''$ of
$K\setminus \{v,v_j\}$ such that $v_j+\sum_{u\in K'}u=v+\sum_{u\in K''}u$. Thus,
$v=v_j+\sum_{u\in K'\cup K''}u$, a contradiction to the definition of $K.$\qed
\end{proof}

\subsection{Max-$r$-Lin2-AA}\label{sec:maxrlin}

The following result was proved by Gutin et al. \cite{GutinKimSzeiderYeo09a}.

\begin{theorem}\label{thm:Max-r-Lin2fpt}
 The problem {\sc Max-$r$-Lin2-AA} admits a kernel with at most $O(k^2)$ variables and equations.
\end{theorem}
\begin{proof}
Let the system $S$ be irreducible. Consider the excess
\begin{equation}\label{eq:excess} \varepsilon(x)=\sum_{I\in \cal F}w_I(-1)^{b_I}\prod_{i\in I}x_i. \end{equation}
Let us assign value $-1$ or $1$ to each $x_i$ with probability $1/2$ independently of the other variables.
Then $X=\varepsilon(x)$ becomes a random variable.
By Lemma \ref{lem:Pars}, we have $\mathbb{E}(X^2)=\sum_{I\in \cal F}w^2_I$.
Therefore, by Lemmas \ref{lem32} and \ref{lem41},  $$\Prob[\ X\ge \sqrt{m}/(2\cdot 3^r)\ ]\ge
\Prob\left[\ X\ge \sqrt{\sum_{I\in \cal F}w^2_I}/(2\cdot 3^r)\ \right]>0.$$ Hence by Remark \ref{MaxLinExcess}, if
$\sqrt{m}/(2\cdot 3^r)\ge 2k$, then
the answer to \textsc{Max-$r$-Lin2-AA} is {\sc
Yes}. Otherwise, $m=O(k^2)$ and, by  Rule \ref{rulerank}, we have $n\le m=O(k^2)$.\qed
\end{proof}

The bound on the number of variables can be improved and it was done by Crowston et al. \cite{CrowstonSWAT} and Kim and Williams
\cite{KimWil}. The best known improvement is by Crowston, Fellows et al. \cite{CroFelGut}:

\begin{theorem}\label{thm:linkernel4lin2}
The problem {\sc Max-$r$-Lin2-AA} admits a kernel with at most $(2k-1)r$ variables.
\end{theorem}

This theorem can be easily proved using Formula (\ref{eq:excess}), Lemma \ref{lem:M-freeapplic} and the following result by Crowston, Fellows et al. \cite{CroFelGut}.

\begin{lemma}
Let $M$ be a set of vectors in $\mathbb{F}^n_2$ such that $M$ contains a basis of $\mathbb{F}^n_2.$
Suppose that each vector of $M$ contains at most $r$ non-zero coordinates. If $k\ge 1$ is an integer and $n \ge r(k-1)+1$,
then in time $|M|^{O(1)}$, we can find a subset $K$ of $M$ of $k$ vectors
such that $K$ is $M$-sum-free.
\end{lemma}

Both Theorem \ref{thm:linkernel4lin2} and a slightly weaker analogous result of \cite{KimWil} imply the following:

\begin{corollary} There is an algorithm of runtime $2^{O(k)}+m^{O(1)}$ for {\sc Max-$r$-Lin2-AA}.
\end{corollary}

 Kim and Williams \cite{KimWil} proved that the last result is best possible, in a sense, if the Exponential Time Hypothesis holds.

\begin{theorem}\cite{KimWil}
If {\sc Max-3-Lin2-AA} can be solved in $O(2^{\epsilon k}2^{\epsilon m})$ time for every $\epsilon > 0,$ then
3-SAT can be solved in $O(2^{\delta n})$ time for every $\delta > 0,$ where $n$ is the number of variables.
\end{theorem}

\subsection{Max-$r$-CSPs AA}\label{subsec:max-r-csp}

Consider first a detailed formulation of {\sc Max-$r$-CSP-AA}.
Let $V=\{v_1,\ldots ,v_n\}$ be a set of variables, each taking
values $-1$ ({\sc True}) and $1$ ({\sc False}). We are given a set $\Phi$ of Boolean
functions, each involving at most $r$ variables, and a collection ${\cal F}$ of $m$ Boolean functions, each $f \in \cal F$ being a
member of $\Phi$, each with a positive integral weight and each acting on some subset of
$V$. We are to decide whether there is a
truth assignment to the $n$ variables such that the total weight of satisfied functions is at least $A+k$,
where $A$ is the average weight (over all truth assignments) of satisfied functions and $k$ is the parameter.

Note that $A$ is a tight lower
bound for the problem, whenever the family ${\Phi}$ is closed under
replacing each variable by its complement, since if we apply any
Boolean function to all $2^r$ choices of literals whose underlying
variables are any fixed set of $r$ variables, then any truth
assignment  to the variables satisfies exactly the same number of
these $2^r$ functions.

Note that if $\Phi$ consists of clauses, we get {\sc Max-$r$-Sat-AA}. In {\sc Max-$r$-Sat-AA}, $A=\sum_{j=1}^mw_j(1-2^{-r_j}),$
where $w_j$ and $r_j$ are the weight and the number of variables of Clause $j$, respectively.
Clearly, $A$ is a tight lower bound for {\sc Max-$r$-Sat}.

Following \cite{AGK04}, for a Boolean
function $f$ of weight $w(f)$ and on $r(f)\le r$ Boolean variables
$x_{i_1},  \ldots , x_{i_{r(f)}},$ we introduce
a polynomial $h_f(x),\ x=(x_1,  \ldots ,x_n)$ as follows. Let $S_f \subset \{-1,1\}^{r(f)}$
denote the set of all
satisfying assignments of $f$. Then
$$
h_f(x) = w(f)2^{r-r(f)}\sum_{(v_1, \ldots ,v_{r(f)}) \in S_f}
[\prod_{j=1}^{r(f)} (1+x_{i_j} v_j) - 1].
$$

Let $h(x)=\sum_{f \in \cal F} h_f(x).$
It is easy to see (cf. \cite{AlonEtAl2009a}) that the value of $h(x)$ at some $x^0$ is precisely $2^r(U-A)$,
where $U$ is the total weight of
the functions satisfied by the truth assignment $x^0$.  Thus, the answer to {\sc Max-$r$-CSP-AA} is {\sc Yes}
if and only if there is a truth assignment $x^0$ such that $h(x^0)\ge k2^r.$

Algebraic simplification of $h(x)$ will lead us the following
(Fourier expansion of $h(x)$, cf. \cite{odonn}):

\begin{equation}\label{foureq} h(x)=\sum_{S\in {\cal F}}c_S\prod_{i\in S}x_i,\end{equation}
where ${\cal F}=\{\emptyset\neq S\subseteq [n]:\ c_S\neq 0, |S|\le r \}$. Thus, $|{\cal F}|\le n^r$.
The sum $\sum_{S\in {\cal F}}c_S\prod_{i\in S}x_i$ can be viewed as the excess of an instance of
{\sc Max-$r$-Lin2-AA} and, thus, we can reduce {\sc Max-$r$-CSP-AA}
into {\sc Max-$r$-Lin2-AA} in polynomial time (since $r$ is fixed, the algebraic simplification can be done in polynomial
time and it does not matter whether the parameter of {\sc Max-$r$-Lin2-AA} is $k$ or $k'=k2^r$). By Theorem \ref{thm:Alon}, {\sc Max-$r$-Lin2-AA} has a kernel with $O(k^2)$ variables and equations. This kernel is a bikernel from {\sc Max-$r$-CSP-AA} to {\sc Max-$r$-Lin2-AA}. Thus, by Lemma \ref{lem:pk}, we obtain the following theorem of Alon et al. \cite{AlonEtAl2009a}.

\begin{theorem}
{\sc Max-$r$-CSP-AA} admits a polynomial-size kernel.
\end{theorem}

Applying a reduction from {\sc Max-$r$-Lin2-AA} to {\sc Max-$r$-Sat-AA} in which each monomial in (\ref{foureq}) is replaced by $2^{r-1}$ clauses, Alon et al. \cite{AlonEtAl2009a} obtained the following:

\begin{theorem}
{\sc Max-$r$-Sat-AA} admits a kernel with $O(k^2)$ clauses and variables.
\end{theorem}

Using also Theorem \ref{thm:linkernel4lin2}, it is easy to improve this theorem with respect to the number of variables in the kernel.
This result was first obtained by Kim and Williams \cite{KimWil}.

\begin{theorem}
{\sc Max-$r$-Sat-AA} admits a kernel with $O(k)$  variables.
\end{theorem}

\section{MaxLin2-AA and MaxSat-AA} \label{sec:maxlin}

Recall that {\sc MaxLin2-AA} is the same problem as {\sc Max-$r$-Lin2-AA}, but the number of variables in an equation is not bounded. Thus, {\sc MaxLin2-AA} is a generalization of {\sc Max-$r$-Lin2-AA}. In this section we present a scheme of a recent proof by Crowston, Fellows et al. \cite{CroFelGut} that {\sc MaxLin2-AA} is fpt and has a kernel with polynomial number of variables. This result finally solved an open question of Mahajan, Raman and Sikdar \cite{MahajanRamanSikdar09}. Still, we do not know whether {\sc MaxLin2-AA} has a kernel of polynomial size and we are able to give only partial results on the topic.

\begin{theorem}\label{thm:main}\cite{CroFelGut}
The problem {\sc MaxLin2-AA} has a kernel with at most $O(k^2\log k)$ variables.
\end{theorem}

The proof of this theorem in \cite{CroFelGut} which we give later is based on Theorems \ref{thm:ing2} and \ref{thm:ing1}.

\begin{theorem}\label{thm:ing2}\cite{CrowstonSWAT}
Let $S$ be an irreducible system of {\sc MaxLin2-AA} and let $k\ge 2.$ If $k\le m\le 2^{n/(k-1)}-2$, then the maximum excess of $S$ is at least $k$. Moreover, we can find an assignment with excess of at least $k$ in time $m^{O(1)}$.
\end{theorem}

This theorem can easily be proved using Lemma \ref{lem:M-freeapplic} and the following lemma.

\begin{lemma}\cite{CrowstonSWAT}
Let $M$ be a set in $\mathbb{F}^n_2$ such that $M$ contains a basis of $\mathbb{F}^n_2$, the
zero vector is in $M$ and $|M|<2^n$. If $k$ is a positive integer and $k+1\le |M|\le 2^{n/k}$
then, in time $|M|^{O(1)}$, we can find an $M$-sum-free subset $K$ of $M$ with at least $k+1$ vectors.
\end{lemma}

\begin{theorem}\label{thm:ing1}\cite{CroFelGut}
There exists an $n^{2k}(nm)^{O(1)}$-time algorithm for {\sc MaxLin2-AA} that returns an
assignment of excess of at least $2k$ if one exists, and returns {\sc No} otherwise.
\end{theorem}

The proof of this theorem in \cite{CroFelGut} is based on constructing a special depth-bounded search tree.

Now we will present the proof of Theorem \ref{thm:ing2} from  \cite{CroFelGut}.

\noindent{\bf Proof of Theorem \ref{thm:ing2}:} Let $\cal L$ be an instance of {\sc MaxLin2-AA} and let $S$ be
the system of $\cal L$ with $m$ equations and $n$ variables.
We may assume that $S$ is irreducible. Let the parameter $k$ be an
arbitrary positive integer.

If $m<2k$ then $n<2k=O(k^2\log k)$. If $2k\le m\le 2^{n/(2k-1)}-2$
then, by Theorem \ref{thm:ing2} and Remark \ref{MaxLinExcess}, the answer to $\cal L$ is {\sc Yes}
and the corresponding assignment can be found in polynomial time. If $m\ge
n^{2k}$ then, by Theorem \ref{thm:ing1}, we can solve $\cal L$ in
polynomial time.

Finally we consider the case $2^{n/(2k-1)}-1 \le m \le n^{2k}-1$.
Hence, $n^{2k}\ge 2^{n/(2k-1)}.$
Therefore, $4k^2\ge 2k+n/\log n \ge \sqrt{n}$ and $n\le (2k)^4$. Hence,
$n\le 4k^2\log n\le 4k^2\log(16k^4)=O(k^2\log k).$

Since $S$ is irreducible, $m<2^n$ and thus we have obtained the desired kernel.
\qed

\vspace{3mm}

Now let us consider some cases where we can prove that {\sc MaxLin2-AA} has a polynomial-size kernel.
Consider first the case when each equation in $S$ has odd number of variables. Then we have the following theorem proved by Gutin et al. \cite{GutinKimSzeiderYeo09a}.

\begin{theorem}\label{thm:odd}
 The special case of {\sc MaxLin2-AA} when each equation in $S$ has odd number of variables, admits a kernel with at most $4k^2$ variables and equations.
\end{theorem}
\begin{proof}
Let the system $S$ be irreducible by Rule \ref{rule1}. Consider the excess $\epsilon(x)=\sum_{I\in \cal F}w_I(-1)^{b_I}\prod_{i\in I}x_i$. Let us assign value $-1$ or $1$ to each $x_i$ with probability $1/2$ independently of the other variables. Then $\epsilon(x)$ becomes a random variable.
Since $\epsilon(-x)=-\epsilon(x)$, $\epsilon(x)$ is a symmetric random variable. Let $X=\epsilon(x)$. By Lemma \ref{lem:Pars}, we have $\mathbb{E}(X^2)=\sum_{i\in I}w^2_I$. Therefore, by Lemma \ref{eqsym},  $\Prob(\ X\ge \sqrt{m}\ )\ge \Prob(\ X\ge \sqrt{\sum_{j=1}^m
w^2_j}\ )>0.$ Hence, if $\sqrt{m}\ge 2k$, the answer to \textsc{MaxLin2-AA} is {\sc
Yes}. Otherwise, $m<4k^2$ and, after applying  Rule \ref{rulerank}, we have $n\le m\le 4k^2$.\qed
\end{proof}

In fact, Gutin et al. \cite{GutinKimSzeiderYeo09a} proved the following more general result.

\begin{theorem}
The following special case of {\sc MaxLin2-AA} admits a kernel with at most $4k^2$ variables and equations:
there exists a subset $U$ of variables such that
each equation in $Ax=b$ has odd number of variables from $U$.
\end{theorem}

Let us turn to results on {\sc MaxLin2-AA} that do not require any parity conditions. One such result is Theorem \ref{thm:Max-r-Lin2fpt}.
Gutin et al. \cite{GutinKimSzeiderYeo09a} also proved the following `dual' theorem.

\begin{theorem}\label{thm:Alon}
Let $\rho\ge 1$ be a fixed integer. Then {\sc MaxLin2-AA} restricted to instances where no variable appears in more than $\rho$ equations,
admits a kernel with $O(k^2)$ variables and equations.
\end{theorem}

The proof is similar to that of Theorem \ref{thm:Max-r-Lin2fpt}, but Lemma \ref{lem:HI2} (in fact, its weaker version obtained in \cite{GutinKimSzeiderYeo09a}) is used instead of Lemma \ref{lem41}.

\vspace{2mm}

Recall that {\sc MaxSat-AA} is the same problem as {\sc Max-$r$-Sat-AA}, but the number of variables in a clause is not bounded. Crowston et al. \cite{CroGutJonRamSau} proved that {\sc MaxSat-AA} is para-NP-complete and, thus, {\sc MaxSat-AA} is not fpt unless P$=$NP.
This is in sharp contrast to {\sc MaxLin2-AA}. This result is a corollary of the following:

\begin{theorem}\cite{CroGutJonRamSau}\label{thm:MaxSatlog}
{\sc Max-$r(n)$-Sat-AA} is para-NP-complete for $r(n)=\lceil \log n\rceil$.
\end{theorem}

The Exponential Time Hypothesis (ETH) claims that {\sc 3-SAT} cannot be solved in time $2^{o(n)}$, where $n$ is the number of variables (see, e.g., \cite{FlumGrohe06,Niedermeier06}). Using ETH, we can improve Theorem \ref{thm:MaxSatlog}.

\begin{theorem}\cite{CroGutJonRamSau}\label{thm:MaxSatloglog}
Assuming ETH, {\sc Max-$r(n)$-Sat-AA} is not fpt for any $r(n) \ge \log\log n+\phi(n)$, where $\phi(n)$ is any unbounded strictly increasing function of $n$.
\end{theorem}

The following theorem shows that Theorem \ref{thm:MaxSatloglog} provides a bound on $r(n)$ which is not far from optimal.

\begin{theorem}\label{thm:maxsataafpt}\cite{CroGutJonRamSau}
{\sc Max-$r(n)$-Sat-AA} is fpt for $r(n)\le \log\log n-\log\log\log n-\phi(n)$, for any unbounded strictly increasing function $\phi(n)$.
\end{theorem}

\section{Ordering CSPs}\label{sec:perm}

In this section we will discuss recent results in the area of {\em Ordering Constraint
Satisfaction Problems (Ordering CSPs)} parameterized above average. Ordering CSPs include several well-known
problems such as {\sc Betweenness}, {\sc Circular Ordering} and  {\sc Acyclic
Subdigraph} (which is equivalent to \textsc{2-Linear Ordering}). These three problems have applications
in circuit design and computational biology \cite{CS98,Opatrny1979}, in qualitative spatial reasoning \cite{IC00}, and in economics \cite{Rei85}, respectively.

Let us define Ordering CSPs of arity 3. The reader can easily generalize it to any arity $r\ge 2$ and we will do it below for \textsc{Linear Ordering} of arity $r.$
Let $V$ be a set of $n$ variables and let $$\Pi\subseteq \mathcal S_3 = \{(123),(132),(213),(231),(312),(321)\}$$ be arbitrary.
A {\em constraint set over $V$} is a multiset $\mathcal C$ of {\em constraints}, which are permutations of three distinct elements of $V$.
A bijection $\alpha:\ V\rightarrow [n]$ is called an {\em ordering} of $V.$
For an ordering $\alpha:\ V\rightarrow [n]$, a constraint $(v_1,v_2,v_3)\in\mathcal C$ is
{\em $\Pi$-satisfied by $\alpha$} if
there is a permutation $\pi\in \Pi$ such that $\alpha(v_{\pi(1)})<\alpha(v_{\pi(2)})<\alpha(v_{\pi(3)})$.
Thus, given $\Pi$ the problem $\Pi$-CSP, is the problem of deciding if there exists an ordering of $V$  that $\Pi$-satisfies all the constraints.
Every such problem is called an Ordering CSP of arity 3. We will consider the maximization version
of these problems, denoted by {\sc Max-$\Pi$-CSP}, parameterized above the average number
of constraints satisfied by a random ordering of $V$ (which can be shown to be a tight bound).

Guttmann and Maucher \cite{GuttmannMaucher2006} showed that there are in fact only $13$ distinct $\Pi$-CSP's of arity 3 up to symmetry, of which $11$ are nontrivial. They
are listed in Table \ref{tab:allproblems} together with their complexity.
Note that if  $\Pi=\{(123),(321)\}$ then we obtain the {\sc Betweenness} problem and if
$\Pi=\{(123)\}$ then we obtain {\sc 3-Linear Ordering}.

\begin{table}
\centering
  \begin{tabular}{llc}
    $\Pi\subseteq \mathcal S_3$~~                  & ~~Name~~   & ~~Complexity\\
    \noalign{\smallskip}
    \hline
    \noalign{\smallskip}
    $\Pi_0=\{(123)\}$       & ~~{\sc Linear Ordering-3}                       & polynomial\\[0.1cm]
    $\Pi_1 =\{(123), (132)\}$       &  ~~                  & polynomial\\[0.1cm]
    $\Pi_{2} =\{(123), (213), (231)\}$       &  ~~                  & polynomial\\[0.1cm]
    $\Pi_{3} =\{(132),  (231), (312), (321)\}$       &  ~~                  & polynomial\\[0.1cm]
    $\Pi_4 = \{(123),(231)\}$                      & ~~                        & $\mathsf{NP}$-comp.\\[0.1cm]
    $\Pi_5 = \{(123),(321)\}$                      & ~~{\sc Betweenness}       & $\mathsf{NP}$-comp.\\[0.1cm]
    $\Pi_6 = \{(123),(132),(231)\}$                & ~~                        & $\mathsf{NP}$-comp.\\[0.1cm]
    $\Pi_7 = \{(123),(231),(312)\}$                & ~~{\sc Circular Ordering} & $\mathsf{NP}$-comp.\\[0.1cm]
    $\Pi_8 = \mathcal S_3\setminus\{(123),(231)\}$ & ~~                        & $\mathsf{NP}$-comp.\\[0.1cm]
    $\Pi_9 = \mathcal S_3\setminus\{(123),(321)\}$ & ~~{\sc Non-Betweenness}   & $\mathsf{NP}$-comp.\\[0.1cm]
    $\Pi_{10} = \mathcal S_3\setminus\{(123)\}$       & ~~                        & $\mathsf{NP}$-comp.\\[0.1cm]
  \end{tabular}
  \vspace{0.2cm}
  \caption{Ordering CSPs of arity 3 (after symmetry considerations)}
\label{tab:allproblems}
\end{table}

Gutin et al. \cite{GutinIerselMnichYeo} proved that all $11$ nontrivial {\sc Max-$\Pi$-CSP} problems are NP-hard
(even though four of the {\sc $\Pi$-CSP} are polynomial).

Now observe that given a variable set $V$ and a constraint multiset $\mathcal C$ over $V$, for a random ordering $\alpha$ of $V$, the probability of a constraint in $\mathcal C$ being $\Pi$-satisfied by $\alpha$ equals $\frac{|\Pi|}{6}$.
Hence, the expected number of satisfied constraints from $\mathcal C$ is $\frac{|\Pi|}{6}|\mathcal C|$, and thus there is an ordering $\alpha$ of $V$ satisfying at least $\frac{|\Pi|}{6}|\mathcal C|$ constraints (and this bound is tight). A derandomization argument leads to $\frac{|\Pi_i|}{6}$-approximation algorithms for the problems {\sc Max-$\Pi_i$-CSP} \cite{ChaGurMan}. No better constant factor approximation is possible assuming the Unique Games Conjecture~\cite{ChaGurMan}.

We will study the parameterization of {\sc Max-$\Pi_i$-CSP} above tight lower bound:

\medskip
\noindent \begin{tabular}{lp{0.85\textwidth}}
\multicolumn{2}{l}{{\sc $\Pi$-Above Average ($\Pi$-AA)}}   \\
\textit{Input:}     & A finite set $V$ of variables, a multiset $\mathcal C$ of ordered trip\-les  of distinct variables from $V$ and an integer $\kappa\geq 0$.                   \\
\textit{Parameter:} & $\kappa$.\\
\textit{Question:}  & Is there an ordering $\alpha$ of $V$ such that at least $\frac{|\Pi|}{6}|\mathcal C| + \kappa$ constraints of $\mathcal C$ are $\Pi$-satisfied by $\alpha$?\\
\end{tabular}
\medskip

In \cite{GutinIerselMnichYeo} it is shown that all $11$ nontrivial {\sc $\Pi$-CSP-AA} problems admit kernels with $O(\kap{}^2)$ variables.
This is shown by first reducing them to {\sc 3-Linear Ordering-AA} (or {\sc 2-Linear Ordering-AA}), and then finding a kernel for this problem,
which is transformed back to the original problem. The first transformation is easy due to the following:

\begin{proposition}\cite{GutinIerselMnichYeo}
\label{thm:allreducetoone}
  Let $\Pi$ be a subset of $\mathcal S_3$ such that $\Pi\notin\{\emptyset,\mathcal S_3\}$.
  There is a polynomial time transformation $f$ from {\sc $\Pi$-AA} to {\sc 3-Linear Ordering-AA} such that an instance $(V,{\mathcal C},k)$ of {\sc $\Pi$-AA} is a {\sc Yes}-instance if and only if $(V,{\mathcal C}_0,k)=f(V,{\mathcal C},k)$ is a {\sc Yes}-instance of {\sc 3-Linear Ordering-AA}.
\end{proposition}
\begin{proof}
  From an instance $(V,\mathcal C,k)$ of~{\sc $\Pi$-AA}, construct an instance $(V,\mathcal C_0,k)$ of {\sc 3-Linear Ordering-AA} as follows.
  For each triple $(v_1,v_2,v_3)\in\mathcal C$, add $|\Pi|$ triples $(v_{\pi(1)},v_{\pi(2)},v_{\pi(3)})$, $\pi\in \Pi$, to~$\mathcal C_0$.

  Observe that a triple $(v_1,v_2,v_3)\in\mathcal C$ is $\Pi$-satisfied if and only if exactly one of the triples $(v_{\pi(1)},v_{\pi(2)},v_{\pi(3)})$, $\pi\in \Pi$, is satisfied by {\sc 3-Linear Ordering}.
  Thus, $\frac{|\Pi|}{6}|\mathcal C| + k$ constraints from $\mathcal C$ are $\Pi$-satisfied if and only if the same number of constraints from $\mathcal C_0$ are satisfied by {\sc 3-Linear Ordering}.
  It remains to observe that $\frac{|\Pi|}{6}|\mathcal C| + k=\frac{1}{6}|\mathcal C_0| + k$ as $|\mathcal C_0|=|\Pi|\cdot |\mathcal C|$.\qed
\end{proof}

Recall that the maximization version of \textsc{$r$-Linear Ordering} ($r \geq 2$) can be defined as follows.
An instance of such a problem
consists of a set of variables $V$ and a multiset of constraints, which are
ordered $r$-tuples of distinct variables of $V$ (note that the same set of $r$ variables
may appear in several different constraints). The objective is
to find an ordering $\alpha$ of $V$ that maximizes
the number of constraints whose order in $\alpha$
follows that of the constraint (we say that these constraints are satisfied). It is well-known that \textsc{2-Linear Ordering} is NP-hard (it follows immediately from the fact proved by Karp \cite{karp72} that the feedback arc set problem is NP-hard). It is easy to extend this hardness result to all \textsc{$r$-Linear Ordering} problems (for each fixed $r\ge 2$). Note that in  \textsc{$r$-Linear Ordering Above Average} (\textsc{$r$-Linear Ordering-AA}), given a multiset $\mathcal C$ of constraints over $V$ we are to decide whether there is an ordering of $V$ that satisfies at least $|{\mathcal C}|/r!+\kappa$ constraints.

\textsc{(2,3)-Linear Ordering} is a mixture of \textsc{2-Linear Ordering} and \textsc{3-Linear Ordering}, where constraints can be of both arity 2 and 3.

We proceed by first considering \textsc{2-Linear Ordering} (Subsection \ref{subsec:2LO}), {\sc Betweenness} (Subsection \ref{subsec:btw}), and
\textsc{3-Linear Ordering} (Subsection \ref{subsec:3LO}) separately and proving the existence of a kernel with a quadratic number of variables and constraints for their parameterizations above average. We will conclude the section by briefly overviewing the result of Kim and Williams \cite{KimWil} that \textsc{(2,3)-Linear Ordering} has a kernel with a linear number of variables (Subsection \ref{subsec:23LO}). By considering \textsc{(2,3)-Linear Ordering} rather than just \textsc{3-Linear Ordering} separately, Kim and Williams managed to obtain a finite set of reduction rules which appear to be impossible to obtain for \textsc{3-Linear Ordering} only (see Subsection \ref{subsec:3LO}).

\subsection{2-Linear Ordering}\label{subsec:2LO}
\newcommand{\dom}{\rightarrow}

Let $D=(V,A)$ be a digraph on $n$ vertices with no loops or parallel arcs in which every arc
$ij$ has a positive integral weight $w_{ij}$. Consider an
ordering $\alpha: V \dom [n]$ and the subdigraph $D_{\alpha}=(V,\{
ij\in A:\ \alpha(i)<\alpha(j)\})$ of $D$. Note that $D_{\alpha}$ is acyclic.
The problem of finding a subdigraph $D_{\alpha}$ of $D$ of maximum weight is equivalent to \textsc{2-Linear Ordering}
(where the arcs correspond to constraints and weights correspond to the number of occurrences of each constraint).

It is easy to see that, in the language of digraphs, \textsc{2-Linear Ordering-AA} can be formulated as follows.

\begin{quote}
  \textsc{2-Linear Ordering Above Average} (\textsc{2-Linear Ordering-AA})\nopagebreak

  \emph{Instance:} A digraph $D=(V,A)$, each arc $ij$ has an integral positive
  weight $w_{ij}$, and a positive integer~$\kap{}$.\nopagebreak

  \emph{Parameter:} The integer $\kap{}$.\nopagebreak

  \emph{Question:} Is there a subdigraph $D_{\alpha}$ of $D$ of weight at least
  $W/2+\kap{}$, where $W=\sum_{ij\in A}w_{ij}$ ?
\end{quote}

Mahajan, Raman, and Sikdar~\cite{MahajanRamanSikdar09} asked whether
\textsc{2-Linear Ordering-AA} is fpt for the special case when all arcs
are of weight 1. Gutin et al. \cite{GutinKimSzeiderYeo09a} solved the problem by obtaining a
quadratic kernel for the problem. In fact, the problem can be solved using the following result of Alon \cite{Alon2002}:
there exists an ordering $\alpha$ such that $D_{\alpha}$ has weight at least $(\frac{1}{2}+\frac{1}{16|V|})W.$
However, the proof in \cite{Alon2002} uses a probabilistic approach for which a derandomization is not known yet and, thus, we cannot
find the appropriate $\alpha$ deterministically. Moreover, the probabilistic approach in \cite{Alon2002} is quite specialized. Thus, we
briefly describe a solution from Gutin et al. \cite{GutinKimSzeiderYeo09a} based on Strictly-Above-Below-Expectation Method (introduced in \cite{GutinKimSzeiderYeo09a}).

Consider the following reduction rule:
\begin{krule}\label{LO1}
Assume $D$ has a directed
\hbox{2-cycle $iji$};
if $w_{ij}=w_{ji}$ delete the cycle,
if $w_{ij}>w_{ji}$ delete the arc $ji$ and replace $w_{ij}$ by $w_{ij}-w_{ji}$,
and if $w_{ji}>w_{ij}$ delete the arc $ij$ and replace $w_{ji}$ by $w_{ji}-w_{ij}$.
\end{krule}
It is easy to check that the answer to {2-Linear Ordering-AA} for a digraph $D$ is {\sc Yes} if
and only if the answer to {2-Linear Ordering-AA} is {\sc Yes} for a digraph obtained from $D$
using the reduction rule as long as possible. A digraph is called an {\em oriented graph} if it has no directed 2-cycle.
Note that applying Rule \ref{LO1} as long as possible results in an oriented graph.

Consider a random ordering: $\alpha: V \dom [n]$ and a
random variable $X(\alpha)=\frac{1}{2}\sum_{ij\in A} x_{ij}(\alpha)$, where
$x_{ij}(\alpha)=w_{ij}$ if $\alpha(i)<\alpha(j)$ and
$x_{ij}(\alpha)=-w_{ij}$, otherwise. It is easy to see that $X(\alpha)=\sum\{
w_{ij}:\ ij\in A, \alpha(i)<\alpha(j)\} - W/2$. Thus, the answer to {2-Linear Ordering-AA} is
{\sc Yes} if and only if there is an ordering $\alpha: V \dom [n]$
such that $X(\alpha)\ge \kap{}$. Since $\mathbb{E}(x_{ij})=0$, we have
$\mathbb{E}(X)=0$.

Let $W^{(2)}=\sum_{ij\in A}w_{ij}^2$. Gutin et al. \cite{GutinKimSzeiderYeo09a} proved the following:
\begin{lemma}\label{lemEX2}
If $D$ is an oriented graph, then $\mathbb{E}(X^2)\ge W^{(2)}/12$.
\end{lemma}

Since $X(-\alpha)=-X(\alpha)$, where $-\alpha(i)=n+1-\alpha(i),$ $X$ is a symmetric random variable
and, thus, we use a proof similar to that of Theorem \ref{thm:odd} (but applying Lemma \ref{lemEX2} instead of Lemma \ref{lem:Pars}) to show
the following:

\begin{theorem}\cite{GutinKimSzeiderYeo09a}
{2-Linear Ordering-AA} has a kernel with $O(\kap{}^2)$ arcs.
\end{theorem}

By deleting isolated vertices (if any), we can obtain a kernel with $O(\kap{}^2)$ arcs and vertices. Kim and Williams \cite{KimWil} proved that \textsc{2-Linear Ordering} has a kernel with a linear number of variables.

\subsection{Betweenness}\label{subsec:btw}

Let $V=\{v_1, \ldots,v_n\}$ be a set of variables and let ${\cal C}$ be a multiset of $m$ {\em betweenness} constraints of the form $(v_i,\{v_j, v_k\})$.
For an ordering $\alpha:\ V\rightarrow [n]$, a constraint $(v_i,\{v_j,v_k\})$ {\em is satisfied} if either $\alpha(v_j)< \alpha(v_i) < \alpha(v_k)$ or $\alpha(v_k)< \alpha(v_i) < \alpha(v_j)$.
In the {\sc Betweenness} problem, we are asked to find an ordering $\alpha$ satisfying the maximum number of constraints in ${\cal C}$.
{\sc Betweenness} is NP-hard as even the problem of deciding whether all betweenness constraints in ${\cal C}$ can be satisfied by an ordering $\alpha$ is NP-complete \cite{Opatrny1979}.

Let $\alpha:\ V\rightarrow [n]$ be a random ordering and observe that the probability of a constraint in $\cal C$ to be satisfied is $1/3.$ Thus,
the expected number of satisfied constraints is $m/3$. A triple of betweenness constraints of the form $(v,\{u, w\}),(u,\{v, w\}),(w,\{v, u\})$ is called
a {\em complete triple}. Instances of {\sc Betweenness} consisting of complete triples demonstrate that $m/3$ is a tight lower bound on the maximum number
of constraints satisfied by an ordering $\alpha$. Thus, the following parameterization is of interest:

\begin{quote}
  \textsc{Betweenness Above Average} (\textsc{Betweenness-AA})\nopagebreak

   \emph{Instance:} A multiset $\mathcal C$ of $m$ betweenness constraints over variables $V$ and an integer $\kap{} \geq 0$.\nopagebreak

  \emph{Parameter:} The integer $\kap{}$.\nopagebreak

  \emph{Question:} Is there an ordering $\alpha:V\rightarrow [n]$
   that satisfies at least $m/3+\kap{}$ constraints from $\mathcal C$?
\end{quote}

In order to simplify instances of {\sc Betweenness-AA} we introduce the following reduction rule.

\begin{krule}\label{Brule}
If $\cal C$ has a complete triple, delete it from $\cal C.$ Delete from $V$ all variables that appear only in the deleted triple.
\end{krule}

Benny Chor's question (see \cite[p. 43]{Niedermeier06}) to determine the parameterized complexity of \textsc{Betweenness-AA} was solved
by Gutin et al. \cite{GutinKimMnichYeo} who proved that \textsc{Betweenness-AA} admits a kernel with $O(\kap{}^2)$ variables and constraints (in fact,
\cite{GutinKimMnichYeo} considers only the case when $\cal C$ is a set, not a multiset, but the proof for the general case is the same \cite{GutinIerselMnichYeo}).
Below we briefly describe the proof in \cite{GutinKimMnichYeo}.

Suppose we define a random variable $X(\alpha)$ just as we did for \textsc{2-Linear Ordering}. However such a variable is not symmetric
and therefore we would need to use Lemma
\ref{lem:Pars} on $X(\alpha)$. The problem is that $\alpha$ is a permutation and in Lemma
\ref{lem:Pars} we are looking at polynomials, $f=f(x_1,x_2\ldots,x_n)$, over variables $x_1,\ldots,x_n$ each with domain $\{-1,1\}$.
In order to get around this problem the authors of \cite{GutinKimMnichYeo} considered a different random variable $g(Z)$, which they defined as follows.

Let $Z=(z_1,z_2,\ldots,z_{2n})$ be a set of $2n$ variables with domain $\{-1,1\}$. These $2n$ variables correspond to $n$ variables $z_1^*,z_2^*,\ldots,z_n^*$
such that $z_{2i-1}$ and $z_{2i}$ form the binary representation of $z_i^*$. That is, $z_{i}^*$ is $0$, $1$, $2$ or $3$ depending on the value of
$(z_{2i-1},z_{2i}) \in \{ (-1,-1), (-1,1), (1,-1), (1,1) \}$. An ordering: $\alpha: V \dom [n]$ complies with $Z$ if for every $\alpha(i)<\alpha(j)$
we have $z_i^* \leq z_j^*$. We now define the value of $g(Z)$ as the average number of constraints satisfied over all orderings which comply with $Z$.
Let $f(Z)=g(Z) - m/3$, and by Lemma \ref{lem:1} we can now use Lemma \ref{lem:Pars} on $f(Z)$ as it is a polynomial over variables whose domain is $\{-1,1\}$.
We consider variables $z_i$ as independent uniformly distributed random variables and then $f(Z)$ is also a random variable.
In \cite{GutinKimMnichYeo} it is shown that the following holds if Reduction Rule \ref{Brule} has been exhaustively applied.

\begin{lemma}\label{lem:1}
The random variable $f(Z)$ can be expressed as a polynomial of degree 6.
We have $\mathbb{E}[f(Z)]=0$. Finally, if $f(Z)\ge \kap{}$ for some $Z\in \{-1,1\}^{2n}$ then the corresponding instance of {\sc Betweenness-AA} is a {\sc Yes}-instance.
\end{lemma}

\begin{lemma} \label{lem:2} \cite{GutinIerselMnichYeo}
For an irreducible (by Reduction Rule \ref{Brule}) instance we have $\mathbb{E}[f(Z)^2] \geq \frac{11}{768}m$.
\end{lemma}

\begin{theorem}\label{thmain} \cite{GutinIerselMnichYeo}
{\sc Betweenness-AA} has a kernel of size $O(\kap{}^2).$
\end{theorem}
\begin{proof}
Let $(V,\mathcal C)$ be an instance of {\sc Betweenness-AA}. We can obtain an irreducible instance $(V',\mathcal C')$ such that
$(V,\mathcal C)$ is a {\sc Yes}-instance if and only if $(V',{\mathcal C}')$ is a {\sc Yes}-instance in polynomial time.
Let $m'=|{\mathcal C}'|$ and let  $f(Z)$ be the random variable defined above. Then $f(Z)$ is expressible as a polynomial of degree 6 by Lemma \ref{lem:1};
hence it follows from Lemma~\ref{lem41}  that $\mathbb{E}[f(Z)^4] \leq 2^{36} \mathbb{E}[f(Z)^2]^2$.
Consequently, $f(Z)$ satisfies the conditions of Lemma~\ref{lem32},
from which we conclude that $\mathbb P\left(f(Z) > \frac{1}{4\cdot 2^{18}} \sqrt{\frac{11}{768}m'}\right) > 0$, by Lemma \ref{lem:2}.
Therefore, by Lemma \ref{lem:1}, if $\frac{1}{4\cdot 2^{18}} \sqrt{\frac{11}{768}m'}\ge \kap{}$ then $(V',{\mathcal C}')$ is a {\sc Yes}-instance for {\sc Betweenness-AA}.
Otherwise, we have $m'= O(\kap{}^2)$.
This concludes the proof of the theorem.
\qed
\end{proof}

By deleting variables not appearing in any constraint, we obtain a kernel with $O(\kap{}^2)$ constraints and variables.

\subsection{3-Linear Ordering}\label{subsec:3LO}

In this subsection, we will give a short overview of the proof in \cite{GutinIerselMnichYeo} that \textsc{3-Linear Ordering} has a kernel with at most $O(\kap{}^2)$ variables and constraints.

Unfortunately, approaches which we used for {\sc 2-Linear Ordering-AA} and {\sc Betweenness-AA} do not work for this problem. In fact, if we wanted to remove subsets of constraints where only the average number of constraints can be satisfied such that after these removals we are guaranteed to have more than the average number of constraints satisfied, then, in general case, an {\em infinite} number
of reduction rules would be needed. The proof of this is quite long and therefore omitted from this survey, see \cite{GutinIerselMnichYeo}
for more information.

However, we can reduce an instance of {\sc 3-Linear Ordering-AA} to instances of {\sc Betweenness-AA} {\em and}
{\sc 2-Linear Ordering-AA} as follows.
With an instance $(V,{\mathcal C})$ of {\sc 3-Linear Ordering-AA}, we associate an instance $(V,\mathcal B)$ of {\sc Betweenness-AA} and two
instances $(V,A')$ and $(V,A'')$ of {\sc 2-Linear Ordering-AA} such that if $C_p=(u,v,w)\in {\mathcal C}$, then add
$B_p=(v,\{u,w\})$ to $\mathcal B$, $a'_p=(u,v)$ to $A'$, and $a''_p=(v,w)$ to $A''$.

Let $\alpha$ be an ordering of $V$ and let $\mathsf{dev}(V,{\mathcal C},\alpha)$ denote the number of constraints satisfied by $\alpha$ minus the average number of satisfied constraints in $(V,{\mathcal C})$, where $(V,{\mathcal C})$ is an instance of {\sc 3-Linear Ordering-AA}, {\sc Betweenness-AA}  or {\sc 2-Linear Ordering-AA}.

\begin{lemma}\label{lem:dev} \cite{GutinIerselMnichYeo}
  Let $(V,C,\kap{})$ be an instance of {\sc 3-Linear Ordering-AA} and let $\alpha$ be an ordering of $V$.
  Then
  \begin{equation*}
    \mathsf{dev}(V,{\mathcal C},\alpha)
      = \frac{1}{2}\left[   \mathsf{dev}(V,A',\alpha)
                          + \mathsf{dev}(V,A'',\alpha)
                          + \mathsf{dev}(V,{\mathcal B},\alpha)
                  \right].
  \end{equation*}
\end{lemma}

 Therefore, we want to find an ordering satisfying as many constraints as possible from both of our new
type of instances (note that we need to use the same ordering for all the problems).

Suppose we have a {\sc No}-instance of {\sc 3-Linear Ordering-AA}. As above, we replace it by three instances of {\sc Betweenness-AA} and
{\sc 2-Linear Ordering-AA}.
Now we apply the reduction rules for {\sc Betweenness-AA} and {\sc 2-Linear Ordering-AA} introduced above
as well as the proof techniques described in the previous sections in order to show that
the total number of variables and constraints left in any of our instances is bounded by $O(\kap{}^2)$. We then transform these reduced instances
back into an instance of {\sc  3-Linear Ordering-AA} as follows. If $\{v,\{u,w\}\}$ is a {\sc Betweenness} constraint then we add
the {\sc  3-Linear Ordering-AA} constraints $(u,v,w)$ and $(w,v,u)$ and if $(u,v)$ is an {\sc 2-Linear Ordering-AA} constraint
then we add the {\sc  3-Linear Ordering-AA} constraints $(u,v,w)$, $(u,w,v)$ and $(w,u,v)$ (for any $w \in V$). As a result,
we obtain a kernel of {\sc 3-Linear Ordering-AA} with at most $O(\kap{}^2)$ variables and constraints.

\subsection{(2,3)-Linear Ordering-AA}\label{subsec:23LO}

In the previous subsection, we overviewed a result that  {\sc  3-Linear Ordering-AA} has a kernel  with at most $O(\kap{}^2)$ variables and constraints.
This result has been partially improved by Kim and Williams \cite{KimWil} who showed that {\sc  3-Linear Ordering-AA} has a kernel  with at most $O(\kap{})$ variables. We will now outline their approach, where they considered {\sc  (2,3)-Linear Ordering-AA}. That is, we allow constraints to contain $2$ or $3$ variables. Thus, we can apply the following reduction rules, where $w(e)$ denotes the weight of constraint $e$ (i.e., the number of times $e$
appears in the constraint multiset) and if $e=(u,v,w)$ is a constraint then we denote $u$ by $e(1)$, $v$ by $e(2)$ and $w$ by $e(3)$, and $var(e)$ denotes the variables in $e$.

\begin{description}
\item[Redundancy Rule:] Remove a variable $v$ from $V$ if it does not appear in any constraint.
Remove a constraint $e$ from $C$ if its weight is zero.

\item[Merging Rule:] If $e_1$ and $e_2$ are identical, then replace them by a single constraint of weight $w(e_1) + w(e_2)$.

\item[Cancellation Rule:] If there are two constraints $e_1$, $e_2$ with $|e_1|=|e_2|=2$ and $e_2 = (e_1(2), e_1(1))$, let
$w_{\min} = \min\{w(e_1), w(e_2)\}$ and replace the weights by $w(e_1)=w(e_1) - w_{\min}$ and $w(e_2) = w(e_2)  - w_{\min}$.

\item[Edge Replacement Rule:] If $e_1,e_2,e_3$ are three constraints in $C$ with $var(e_1) = var(e_2) = var(e_3)$ and
such that $e_2 = (e_1(2), e_1(1), e_1 (3))$ and $e_3 = (e_1(1),e_1(3),e_1(2))$, then:

\begin{itemize}
\item replace the weight of a constraint by $w(e_i) = w(e_i) - w_{min}$ for each $i = 1, 2, 3$, where $w_{\min} =
\min\{w(e_1), w(e_2), w(e_3)\}$.

\item add the binary ordering constraint $(e_1(1), e_1(3))$ of weight $w_{\min}$.
\end{itemize}

\item[Cycle Replacement Rule:] If $e_1, e_2, e_3$ are three constraints in $C$ with $var(e_1) = var(e_2) = var(e_3)$ and
such that $e_2 = (e_1(2), e_1(3), e_1(1))$ and $e_3 = (e_1(3), e_1(1), e_1(2))$, then:

\begin{itemize}
\item  replace the weight of a constraint by $w(e_i) = w(e_i) - w_{\min}$ for each $i = 1, 2, 3$, where $w_{\min} =
\min\{w(e_1), w(e_2), w(e_3)\}$.

\item  add the three binary ordering constraints $(e_1(1), e_1(2))$, $(e_1(2), e_1(3))$ and $(e_1(3), e_1(1))$, each of weight
$w_{\min}$.
\end{itemize}
\end{description}

In \cite{KimWil} it is shown that these reduction rules produce equivalent instances. In \cite{KimWil} the following theorem is then proved.

\begin{theorem} \label{KimWilMain} \cite{KimWil}
Let $I = (V, C, \kap{})$ be an irreducible (under the above reduction rules)
instance of {\sc (2,3)-Linear Ordering-AA}. If $I$ is a {\sc No}-instance (that is,
less than $\rho W + \kap{}$ constraints in $I$ can be simultaneously satisfied, where $\rho W$ is the average weight of clauses satisfied
by a random ordering), then the number of variables in $I$ is $O(\kap{})$.
\end{theorem}

In order to prove this theorem some above-mentioned techniques were used. Let $n=|V|.$ As for {\sc Betweenness-AA} (see Subsection \ref{subsec:btw}),
Kim and Williams \cite{KimWil} introduced a random variable $f(y_1,\ldots ,y_{2n})$, which is a polynomial of degree 6 with $2n$ random uniformly distributed and independent variables $y_i$, each taking value $1$ or $-1$. The key property of $f(y_1,\ldots ,y_{2n})$ is that for every {\sc No}-instance $I$ we have $f(y_1,\ldots ,y_{2n})<\kappa$ for each $(y_1,\ldots ,y_{2n})\in \{-1,1\}^{2n}.$ In Subsection \ref{subsec:btw}, a similar inequality was used to bound the number of constraints in $I$ using a probabilistic approach. Kim and Williams \cite{KimWil} use a different approach to bound the number of variables in $I$: they algebraically simplify $f(y_1,\ldots ,y_{2n})$ and obtain its Fourier expansion (see (\ref{foureq})). As in Subsection \ref{subsec:max-r-csp}, the Fourier expansion can be viewed as the excess of the corresponding instance of {\sc Max-6-Lin2-AA}. Thus, to bound the number of variables in the Fourier expansion, we can use Theorem \ref{thm:linkernel4lin2} (or, its weaker version obtained in \cite{KimWil}) which implies that the number is $O(\kap{})$.

However, there was a major obstacle that Kim and Williams \cite{KimWil} had to overcome. In general case, as a result of the algebraic simplification, the number of variables in the Fourier expansion may be significantly smaller than $2n$ and, thus, the bound on the number of variables in the Fourier expansion may not be used to bound $n$. To overcome the obstacle, Kim and Williams carefully analyzed the coefficients in the Fourier expansion and established that every variable of $V$ is ``represented'' in the Fourier expansion. As a result, they concluded $I$ can have only $O(\kap{})$ variables.

\medskip

\paragraph{Acknowledgments}
Research of Gutin was supported
in part by the IST Programme of the European Community, under the
PASCAL 2 Network of Excellence. Research of Gutin and Yeo was partially supported by an International Joint grant of Royal Society.

\urlstyle{rm}

\end{document}